\documentclass[a4paper]{article}

\usepackage[shortlabels]{enumitem}
\usepackage[]{geometry}
\usepackage[utf8]{inputenc}
\usepackage{fancybox, calc}
\usepackage{amssymb}
\usepackage{amsmath}
\usepackage{hyperref}
 \usepackage{systeme}
\usepackage{graphicx}
\usepackage{enumitem}
\usepackage{ucs}
\usepackage{amsfonts}
\usepackage{amssymb}
\usepackage{makeidx}
\usepackage{diagbox}
\usepackage{tabularx}
\usepackage{float}
\usepackage[usenames, dvipsnames]{color}
\usepackage{url}
\usepackage{ marvosym }
\usepackage{ wasysym }
\usepackage{ mathrsfs }
\usepackage[utf8]{inputenc}
\usepackage[english]{babel}
\usepackage[normalem]{ulem}
\usepackage{ucs}
\usepackage{amsfonts}
\usepackage{dcolumn}
\usepackage{makeidx}
\usepackage{bm}
\usepackage[usenames]{color}
\usepackage{multirow}
\usepackage{graphicx}
\usepackage{hyperref}
\usepackage{tikz}
\usepackage{subfig}
\usepackage{booktabs}
\usepackage{fancyhdr}
\usepackage{amsmath, amsthm, amssymb}
\usepackage{multicol}
\usepackage{yfonts}
\usepackage{multicol}
\usepackage{caption}
\usepackage{cancel}
\usepackage{amssymb,amsfonts,amssymb,bbm}
\usepackage{amsmath}
\usepackage{phaistos,hieroglf,staves,ifsym,marvosym,skull}
\usepackage{tikz}
\usepackage{transparent}
\usepackage{eso-pic}
\usepackage{soul}
\usepackage{lettrine}
\date{\today}

\def\Rset{\mathbb{R}}
\def\pdf{f}
\def\supp{\Omega}
\newcommand{\down}[1]{\mathfrak {D}_{#1}}
\newcommand{\adown}{{\pdf}^\downarrow_{\alpha}}
\newcommand{\up}[1]{\mathfrak {U}_{#1}}
\newcommand{\aup}{{\pdf}^\uparrow_{\alpha}}

\def\Rset{\mathbb{R}}

\def\esssup{\operatorname{esssup}}
\def\argmax{\operatorname*{argmax}}
\def\sign{\operatorname{sign}}
\def\D{\mathcal{D}}

\def\one{\mathbbm{1}}
\newcommand{\minpdf}[3]{g_{#1,#2,#3}}

\def\support{\mathcal{X}}
\def\gauss{g}

\def\pdf{f}

\newcommand{\descort}[2][]{\mathfrak{E}_{#1}\ifthenelse{\isempty{#2}}{}{ {\left[ #2 \right]}}}
\newcommand{\escort}[2][]{\varepsilon_{#1}\ifthenelse{\isempty{#2}}{}{ {\left[ #2 \right]}}}

\newtheorem{theorem}{Theorem}[section]
\newtheorem{corollary}{Corollary}[section]

\newtheorem{lemma}{Lemma}[section]
\newtheorem{definition}{Definition}[section]
\newtheorem{proposition}{Proposition}[section]

\newtheorem{remark}{Remark}[section]

\numberwithin{equation}{section}

\definecolor{rougeG}{rgb}{.76,0,.12}
\definecolor{vertG}{rgb}{.07,.56,.25}

\title{A new pair of transformations and applications to generalized informational inequalities and Hausdorff moment problem	
}
\author{ {\bf Razvan Gabriel Iagar\footnotemark[1]\footnote{\textit{e-mail}: razvan.iagar@urjc.es} and David Puertas-Centeno\footnotemark[2]\footnote{\textit{e-mail}: david.puertas@urjc.es}} \\ \\
	Departamento de Matem\'atica Aplicada, \\ Ciencia e Ingeniería de los Materiales y Tecnología Electrónica\\
	Universidad Rey Juan Carlos \\
	C/~Tulip\'an s/n, 28933 M\'ostoles \\
	Madrid, Spain}

\usepackage{yfonts}
\usepackage{blindtext}
\usepackage[T1]{fontenc}
\usepackage{palatino}

\begin{document}
\date{}
\maketitle

\begin{abstract}
We introduce a pair of transformations, which are mutually inverse, acting on rather general classes of probability densities in $\Rset$. These transformations have the property of interchanging the main informational measures such as $p-$moments, Shannon and Rényi entropies, and Fisher information. We thus apply them in order to establish extensions and generalizations of the Stam and moment-entropy inequalities in a mirrored domain of the entropic indexes. Moreover, with the aid of the two transformations we establish formal solutions to the Hausdorff entropic moment problem by connecting them with the solutions of the standard Hausdorff problem. In addition, we introduce a Fisher-like moment problem and relate it to the standard Hausdorff moment problem.
\end{abstract}

\medskip


\section{Introduction}
	
The establishment of informational inequalities involving moments, entropy functionals or Fisher information measures has been an important subject of research in the framework of information theory. In many of the works dealing with such inequalities, they are established under the assumption that the probability density functions are bounded, see for example~\cite{Lutwak05,Bercher12}. The hypothesis of boundedness is usually assumed in order to ensure that the corresponding functionals and measures are well-defined. However, some probability density functions that are interesting both for theoretical and practical purposes, such as the well-known Beta distribution, could be unbounded at the edges of their supports. For such densities, up to our knowledge, the main informational inequalities given, for example, in~\cite{Lutwak05,Bercher12}, are not applicable.

In order to remove this limitation and extend the range of application of some of the most significant informational inequalities, in this work we introduce and analyze a pair of transformations turning bounded and regular functions into unbounded ones. At the level of informational functionals, these transformations map, on the one hand, the moment of the transformed density in the power Rényi entropy of the original probability density function. On the other hand, the Rényi entropy of the transformed density is mapped into a Fisher-like measure of the original one. This quite unexpected interchange of the main informational measures through our newly introduced transformations allows us to establish a further range of validity of the exponents involved in the Stam and moment-entropy inequalities proved by Lutwak~\cite{Lutwak05} and Bercher~\cite{Bercher12}. These new inequalities, extended thus to what we call the \emph{mirrored domain}, are applicable to unbounded probability density functions, and their minimizers are as well unbounded. Concerning the Crámer-Rao inequality, which in the usual cases is a direct consequence of the moment-entropy and Stam inequalities by simply multiplying them (see \cite{Lutwak05}), it appears that there is no counterpart in the mirrored case. 

Optimal constants and minimizing density functions in the new range of exponents are established as well in the present work, based on the knowledge of optimal constants and minimizers of the classical inequalities mentioned in the previous paragraphs. In particular, we find that the mirrored versions of the informational inequalities in the standard domain are saturated by stretched Gaussian-like probability density functions, but which are (remarkably) divergent at the edges of their supports.

Let us mention here that the extension of the Stam inequality in~\cite{Lutwak05,Bercher12} has been also generalized by Zozor et al.~\cite{Zozor17} to a further domain in which the new minimizing densities are given by a family of less studied special functions called \textit{generalized trigonometric functions}~\cite{Puertas25}, which are connected to a p-Laplacian eigenvalue problem, as studied in~\cite{Shelupsky59, Drabek99, Takeuchi12}. A surprising fact is that, while in the usual domain the extended minimizing densities are essentially given by \textit{generalized cosine functions}, in the mirrored domain they turn to be \textit{generalized sine functions}, which are not directly related to the corresponding generalized cosine (except for the standard sine and cosine functions), in the sense that the generalized sine and cosine functions are only connected through the derivative, but not by a translation of their argument.

As a byproduct of our transformations, we contribute to the study of the \emph{Hausdorff moment problem} by finding that our new transformations relate the formal solution of the entropic Hausdorff problem (for strictly decreasing functions)~\cite{Romera2001} with the formal solutions of the standard Hausdorff problem. Thus, in line with this remark, we establish the equivalence between fixing a finite set of moments and maximizing an entropy functional and fixing a finite set of Rényi (or Tsallis) entropies and maximizing Fisher-like measures.

In order to describe the previous points in a precise mathematical way, we first remind below a number of definitions and establish some notation that will be employed throughout the paper. Let us mention that all the integrals along the paper will be written over $\Rset$ for simplicity, with the convention that, when the integrand is supported on a smaller set, we understand the integral restricted to the (smaller than $\Rset$) support of the integrand.

\medskip

\noindent \textbf{$p-$th absolute moment.} The $p-$th absolute moments of a random variable with probability density function $\pdf$, with $p \geqslant 0$, are defined as
\begin{equation}\label{eq:pabs}
\mu_p[\pdf] = \int_\Rset |x|^p \, \pdf(x) \, dx  = \left\langle  \, |x|^p \, \right\rangle_\pdf,
\end{equation}
when the integral converges. The $p-$typical deviation is then defined, for $p\in(0,\infty)$, as the $p$-th absolute moment to the power $1/p$, and for the extremal cases $p=0$ and $p=\infty$ by taking the corresponding limits, as described below:
\begin{equation*}
\begin{split}
&\sigma_p[\pdf] =\left(\int_\Rset |x|^p \, \pdf(x) \, dx \right)^{\frac{1}{p}}, \quad {\rm for} \ p > 0, \\
&\sigma_0[\pdf] = \lim\limits_{p \to 0} \sigma_p[\pdf]  = \exp\left(\int_\Rset \pdf(x) \, \log|x| \, dx \right), \\
&\sigma_{\infty}[\pdf]=\lim\limits_{p\to\infty}\sigma_p[\pdf]=\esssup\big\{|x| : x\in\Rset, \pdf(x) > 0 \big\}.
\end{split}
\end{equation*}
In other words, the limit $p\to\infty$ of $\sigma_p[\pdf]$ gives the essential supremum of the support of $\pdf$. Note that, when they exist, the variance and the standard deviation are trivially recovered letting $p=2$ for a centered probability density. Although usually $\sigma_p$ is not defined for negative values of the parameter $p$, in this work the range of values $p\in(-\infty,0)$ appears in the previously mentioned mirrored domain of the parameters. We also employ exponential and logarithmic moments in the paper, thus we recall their definitions below. The logarithmic moments are defined as
\begin{equation}\label{eq:sigmaL}
\sigma_p^{(L)}[f]=\int_{\Rset} f(x)|\log |x||^pdx,
\end{equation}
while the exponential moments are given by~\cite{Mnatsakanov2013}
\begin{equation}\label{eq:sigmaE}
\sigma_{p}^{(E)}[f]=\left(\int_{\Rset} e^{-px} f(x)dx\right)^{\frac{1}{p}}=\left\langle e^{-p x}\right\rangle^{\frac 1{p}}_f.
\end{equation}

\medskip

\noindent \textbf{Rényi and Tsallis entropies.} The Rényi and Tsallis entropies of $\lambda-$order, $\lambda\neq1$, of a probability density function $\pdf$ defined on $\Rset$ (or on a Lebesgue measurable subset of $\Rset$) are defined as
\begin{equation*}
R_\lambda[\pdf] = \frac1{1-\lambda} \log\left( \int_\Rset [\pdf(x)]^\lambda \, dx \right), \qquad T_\lambda[\pdf] = \frac1{\lambda-1} \left( 1 - \int_\Rset [\pdf(x)]^\lambda \, dx \right),
\end{equation*}
while for $\lambda=1$ we recover the Shannon entropy
\begin{equation*}
\lim\limits_{\lambda \to 1}R_{\lambda}[\pdf]=\lim\limits_{\lambda \to 1} T_{\lambda}[\pdf]=S[\pdf]=-\int_\Rset \pdf(x) \, \log\pdf(x) \, dx.
\end{equation*}
In the sequel, most of the inequalities we establish involve the quantity called the Rényi entropy power, defined as
$$
N_\lambda[\pdf] = e^{R_{\lambda}[\pdf]} = \left\langle\pdf^{\lambda-1}(x)\right\rangle^\frac1{1-\lambda}_\pdf.
$$
We will also denote by $N[\pdf]=e^{S[\pdf]}$ the power Shannon entropy. Note that the Rényi and Tsallis entropies are one-to-one mapped as follows:
$$
T_{\lambda}[\pdf]=\frac{e^{(1-\lambda) R_{\lambda}[\pdf]} - 1}{1 - \lambda},
$$
and thus, $\argmax{T_\lambda}=\argmax{R_\lambda}.$ Another entropic-like functional which will be useful in this paper has been introduced in~\cite{Nicolas2002}; that is,
\begin{equation}\label{eq:Sp}
\overline S_p[f]=\left[\int_{\Rset}f(x)|\ln\,f(x)|^pdx\right]^{\frac{1}{p}},
\end{equation}

\medskip

\noindent \textbf{$(p,\lambda)-$Fisher information.} An extension of the Fisher information applicable to derivable probability density functions was introduced by Lutwak and Bercher~\cite{Lutwak05, Bercher12, Bercher12a} as follows: given $p>1$ and $\lambda \in \Rset^*$, the $(p,\lambda)-$Fisher information is defined as
\begin{equation}\label{eq:def_FI}
\phi_{p,\lambda}[\pdf]
\: = \: \left(\int_\Rset \left|\pdf(x)^{\lambda-2} \, \frac{d\pdf}{dx}(x)\right|^{p} \pdf(x) \, dx\right)^{\frac1{p\lambda}}
\: = \: \frac1{|\lambda-1|^{\frac1\lambda}} \, \left\langle \left|\frac {d \pdf^{\lambda-1}}{dx}(x)\right|^p \right\rangle_{\!\pdf}^\frac1{p \lambda},
\end{equation}
when $\pdf$ is differentiable on the closure of its support. Observe that the $(2,1)-$Fisher information reduces to the standard Fisher information of a probability density. Another interesting particular case is the $(1,\lambda)-$Fisher information, corresponding to the total variation of $\frac{\pdf^\lambda}{\lambda}$, provided that $\pdf^\lambda$ is a function with bounded variation. Moreover, for any fixed $\lambda\in\Rset^*$, we have
\begin{equation}\label{eq:Linf}
\lim\limits_{p\to\infty}\phi_{p,\lambda}^\lambda=\left\|\left(\frac{\pdf^{\lambda-1}}{\lambda-1}\right)'\right\|_{L^{\infty}(\Rset)},
\end{equation}
provided $\pdf$ is absolutely continuous and the right hand side is finite.
\medskip

\noindent \textbf{Inequalities involving the previous informational quantities.} The \textbf{Stam inequality} and its generalizations with two parameters establish that the product of Rényi entropy power and the $(p,\lambda)-$Fisher information is bounded from below. More precisely, given $p\in[1,+\infty)$ and $\lambda>\frac1{1+p^*}$, the following inequality
\begin{equation}\label{ineq:bip_Stam}
\phi_{p,\lambda}[\pdf] \, N_{\lambda}[\pdf] \: \geqslant \: \phi_{p,\lambda}[\gauss_{p,\lambda}] \, N_{\lambda}[\gauss_{p,\lambda}]\equiv K^{(1)}_{p,\lambda},
\end{equation}
holds true for any absolutely continuous probability density $\pdf$, according to~\cite{Lutwak05, Bercher12a, Zozor17}. In the right hand side of Eq. \eqref{ineq:bip_Stam}, we denote by $\gauss_{p,\lambda}$ the minimizers of the inequality, known as generalized Gaussians \cite{Lutwak05}, $q-$Gaussian~\cite{Bercher12a} or stretched deformed Gaussian~\cite{Zozor17}. These minimizers have the following explicit expression:
\begin{equation}\label{def:g_plambda}
\gauss_{p,\lambda}(x) \, = \, \frac{a_{p,\lambda}}{\exp_\lambda\left( |x|^{p^*} \right)}
\, = \, a_{p,\lambda}\, \exp_{2-\lambda}\left( - |x|^{p^*} \right),
\end{equation}
with $p>1$, $p^*=\frac{p}{p-1}$ the H\"older conjugated of $p$ and $\exp_\lambda$ is the generalized Tsallis exponential
\begin{equation}\label{def:q-exp}
\exp_\lambda(x) = \left( 1 + (1 - \lambda) \, x \right)_+^\frac1{1-\lambda}, \ \ \lambda \ne 1, \qquad \exp_1(x) \: \equiv \: \lim_{\lambda \to 1} \, \exp_\lambda(x) \: = \: \exp(x),
\end{equation}
where $(\cdot)_+ = \max(0,\cdot)$ denotes the positive part. The range of $\lambda$ in these stretched deformed Gaussians is limited to  $\lambda>1-p^*$. Observe that the limit $p \to 1$ entails $p^*\to\infty$ and then $\gauss_{1,\lambda}$ becomes a constant density over a unit length support. In the other limit case, $p \to \infty$, the inequality also holds true taking as definition for $\phi_{p,\lambda}^{\lambda}$ the essential supremum given in \eqref{eq:Linf}, see~\cite{Lutwak04}. In addition, the definition \eqref{def:g_plambda} and the inequality \eqref{ineq:bip_Stam} can be also extended to exponents $p^* \in (0,1)$ (that is, $p \in (-\infty,0)$). Finally, in the special case $p^*=p=0$ and $\lambda>1$, the minimizer of the Stam inequality has the form
$$
\gauss_{0,\lambda} = a_{0,\lambda}(-\log|x|)_+^{\frac1{\lambda-1}}, \quad a_{0,\lambda} = \frac1{2\Gamma\left(\frac{\lambda}{\lambda-1}\right)}.
$$
It is noteworthy to mention that the support of $\gauss_{p,\lambda}$ is
$$
\support_{p,\lambda}=\left\{\begin{array}{ll}\left( - (\lambda-1)^{\frac1{p^*}} \: , \: (\lambda-1)^{\frac1{p^*}}\right), & {\rm for} \ \lambda>1,\\
\Rset, & {\rm for} \ \lambda \leqslant 1,\end{array}\right.
$$
and the normalization constant $a_{p,\lambda}$ is given by
\begin{equation}\label{def_aplambda}
a_{p,\lambda} =  \left\{ \begin{array}{lll}
\displaystyle  \frac{p^*\,|1-\lambda|^\frac1{p^*}}{2\,B\left(\frac1{p^*},\frac\lambda{|1-\lambda|}+\frac{\one_{\Rset_+}(1-\lambda)}{p}\right)},& \text{for}  & \lambda\neq1,\\[7.5mm]
\displaystyle \frac{p^*}{2\Gamma(\frac1{p^*})},&\text{for} &  \lambda=1,
\end{array}\right.
\end{equation}
where $B(\cdot,\cdot)$ denotes the Beta function, $\one_A$ denotes the indicator function of the set $A$, and $\Rset_+=[0,+\infty)$. Let us stress at this point that, due to the radial symmetry, we will restrict throughout the paper for simplicity the definition of $g_{p,\lambda}$ to $x\in[0,\infty)$.

\smallskip

A second inequality, known as the \textbf{moment-entropy inequality}, involves the Rényi power entropy and the moments $\sigma_p$. More precisely, when $p^* \in[0,\infty )$ and $\lambda>\frac1{1+p^*},$ it was proved in~\cite{Lutwak04, Lutwak05, Bercher12} that
\begin{equation}\label{ineq:bip_E-M}
\frac{\sigma_{p^*}[\pdf]}{N_{\lambda}[\pdf]} \, \geqslant \, \frac{\sigma_ {p^*}[\gauss_{p,\lambda}]}{N_{\lambda}[\gauss_{p,\lambda}]}\equiv K^{(0)}_{p,\lambda},
\end{equation}
and the minimizers of Eq. \eqref{ineq:bip_E-M} are the same stretched deformed Gaussian $g_{p,\lambda}$ as for the generalized bi-parametric Stam inequality Eq.~\eqref{ineq:bip_Stam}.

The family of inequalities~\eqref{ineq:bip_Stam} has been extended in the one-dimensional case by decoupling the parameter $\lambda$ of the Rényi entropy power from the second parameter of the $(p,\lambda)-$Fisher information, obtaining thus a family of inequalities with three free parameters. Letting $p \geqslant 1, \; \beta > 0, \; \lambda > 1 - \beta p^*$, the following inequality
\begin{equation}\label{ineq:trip_Stam}
\phi_{p,\beta}[\pdf] \, N_{\lambda}[\pdf] \: \geqslant \: \phi_{p,\beta}[ \minpdf {p}{\beta}{\lambda}] \, N_{\lambda}[ \minpdf {p}{\beta}{\lambda}]\equiv
K^{(1)}_{p,\beta,\lambda}=|1+\beta-\lambda|^{-\frac{1}{\beta}}\bigg(K^{(1)}_{p,\frac{\beta}{1+\beta-\lambda}}\bigg)^{\frac{1}{1+\beta-\lambda}},
\end{equation}
was established in~\cite{Zozor17}, and its minimizers $\minpdf{p}{\beta}{\lambda}$ are given by some special functions.

\section{Up and down transformations: basic properties}\label{sec:defs}

In this section, we introduce our new transformations that are the core of the present work. As we shall see, they are mutually inverse (under some suitable restrictions on the regularity of $\pdf$), consequently leading to a bijection between two spaces of densities, as explained at the end of the section.

\subsection{The down transformation}\label{sec:down}

We begin with \emph{the down transformation}, which applies to decreasing density functions and involves a change of the independent variable mixing both $\pdf$ and its derivative.
\begin{definition}\label{def:down}
Let $\pdf: \supp\longrightarrow \Rset^+$ be a probability density function with $\supp=(x_i,x_f),$ where $-\infty<x_i<x_f\le\infty$, such that $\pdf'(x)<0,\;\forall x\in\supp.$
Then, for $\alpha\in\Rset,$ we define the transformation $\down{\alpha}[\pdf(x)]$ by
\begin{equation}\label{eq:down}
\down{\alpha}[\pdf(x)](s)=\pdf^\alpha(x(s))|\pdf'(x(s))|^{-1},\qquad s'(x)=\pdf^{1-\alpha}(x) |\pdf'(x)|.
\end{equation}
\end{definition}

\noindent For simplicity, we use the abbreviated notation
$$
\pdf^{\downarrow}_\alpha\equiv\down{\alpha}[\pdf].
$$	
Notice that we work only with decreasing density functions for simplicity, but Definition \ref{def:down} can be also applied to increasing densities. We list below several interesting properties related to the down transformation. In addition, when $\pdf$ is a symmetric probability density, the transformation can be defined by symmetry.


%
\begin{remark}[Probability density and supports]
$\adown=\down{\alpha}[\pdf]$ is a probability density with respect to the $s$ variable, since
\begin{equation}\label{eq:preser_prob_down}
\adown(s)ds=\pdf(x)dx
\end{equation}
Moreover, the length of the support of $\adown$  is given by
\begin{equation}
L(\supp^\downarrow_\alpha)=W_0[\adown]=\left\{\begin{array}{ll}
\left|\frac{\pdf^{2-\alpha}(x_i) - \pdf^{2-\alpha}(x_f)}{2-\alpha}\right|,& \alpha\neq 2,
\\[2mm]
\left|\ln\left(\frac{f(x_i)}{f(x_f)}\right)\right|,& \alpha= 2,
\end{array}\right.
\end{equation}
which can be finite or infinite.
\end{remark}
%


%
\begin{remark}[Canonical election]\label{rem:can}
Note that the class of transformations $\down{\alpha}$ are defined up to a translation. Without loss of generality, for $\alpha\neq 2$ we consider
\begin{equation}\label{eq:can_change}
s(x)=\frac{\pdf(x)^{2-\alpha}}{\alpha-2},
\end{equation}
while for $\alpha=2$
\begin{equation}\label{eq:can_change_a2}
s(x)=-\ln\,\pdf(x).
\end{equation}
With this election, for $\alpha\neq 2$ we have
\[s_i=s(x_i)=\frac{\pdf(x_i)^{2-\alpha}}{\alpha-2},\quad s_f=s(x_f)=\frac{\pdf(x_f)^{2-\alpha}}{\alpha-2}.\]
It is worth mentioning that, with this election and for $\alpha<2$, the variable $s$ takes only negative values; that is, the quantity $(\alpha-2)s$ is always positive. For $\alpha=2$, we have
\[s_i=s(x_i)=-\ln\,\pdf(x_i),\quad s_f=s(x_f)=-\ln\,\pdf(x_f).\]
\end{remark}
%


%
\begin{remark}[Supports]\label{rem:downsupp}
Note that when $\alpha<2$ the length of the support $\supp^\downarrow_\alpha$ is finite whenever $\pdf$ is a bounded density, and infinite for probability densities such that $\pdf(x)\rightarrow\infty$ as $x\to x_i$. However for $\alpha\geq 2$ the support turns out to be finite when $\pdf(x_f)>0$ and infinite when $\pdf(x_f)=0.$  Let us also observe that when $\pdf$ approaches a uniform density, then $L(\supp^\downarrow_\alpha)\rightarrow0$, that is, $\adown$ tends to be a Dirac's delta.
\end{remark}
%


%
\begin{remark}[Divergence on the border]
In the case that $\pdf'(x)\rightarrow 0$ when $x\rightarrow x_i$, we observe that $\adown(s)\rightarrow \infty$ when $s\rightarrow s_i$. For probability densities such that $ \frac{\pdf(x)^\alpha}{|\pdf'(x)|}\le K$ for some $K>0$, it follows that $\adown(s)$ is bounded.
\end{remark}
%


%
\begin{remark}[Derivatives and composition]\label{rem:downcomposition}
In order to apply the down transformation twice, we need the function obtained after the first transformation to be monotone. After some direct calculations, we find
\begin{equation}\label{eq:der_down}
\begin{split}
\frac{d {\adown}}{ds} =& -\frac{d [\pdf^\alpha(x(s))(\pdf'(x(s)))^{-1}]}{ds}
\\
=&-\left(\alpha \pdf^{\alpha-1}(x) x'(s)-\pdf(x(s))^\alpha [\pdf'(x(s))]^{-2} \pdf''(x(s)) x'(s)\right)
\\
=&\frac{\pdf^{2\alpha-2}(x)}{f'(x)}\left(\alpha  -\frac{\pdf(x)\pdf''(x)}{[\pdf'(x)]^{2}}\right) ,
\end{split}
\end{equation}
from where the sign of the derivative of $\adown$ depends exclusively of ${\rm sign}\left(\alpha -\frac{\pdf(x)\pdf''(x)}{[\pdf'(x)]^{2}}\right).$ Thus, if there is $K$ such that
$$
\frac{\pdf(x)\pdf''(x)}{[\pdf'(x)]^{2}}\leqslant K, \quad {\rm for \ any} \ x\in\Rset,
$$
then one can apply twice the transformation $\down{\alpha}$ for any
$$
\alpha>\sup\limits_{x\in\Rset}\left(\frac{\pdf(x)\pdf''(x)}{[\pdf'(x)]^{2}}\right).
$$
Some examples of noticeable probability density functions satisfying the latter property are the $q-$exponentials (or simply powers) or the standard exponential densities of the form $C(\lambda)e^{-\lambda|x|}$. In both cases, the quantity $\frac{\pdf(x)\pdf''(x)}{[\pdf'(x)]^{2}}$ is constant.
\end{remark}
We split the class $\mathcal{D}$ of derivable and monotone probability densities on $(x_i,x_f)$ into three different sets:
\begin{equation*}
\begin{split}
&\mathcal D_0=\{f\in\mathcal D: \lim\limits_{x\to x_i}f'(x)=0\},\\
&\mathcal D_1=\{f\in\mathcal D: \lim\limits_{x\to x_i}f'(x)=K\neq0\},\\
&\mathcal D_\infty=\{f\in\mathcal D: \lim\limits_{x\to x_i}f'(x)=-\infty\}.	
\end{split}
\end{equation*}
\begin{remark}
As a consequence of Eq. \eqref{eq:der_down}, given $f\in\mathcal D_0$ and assuming that
\begin{equation}\label{eq:lims}
\alpha\notin\left[\liminf\limits_{x\to x_i}\frac{\pdf(x)\pdf''(x)}{[\pdf'(x)]^{2}},\limsup\limits_{x\to x_i}\frac{\pdf(x)\pdf''(x)}{[\pdf'(x)]^{2}}\right],
\end{equation}
we find on the one hand that $\adown\in\mathcal D_\infty$. On the other hand, if $f\in\mathcal D_1$ and Eq. \eqref{eq:lims} is in force, then $\adown\in\mathcal D_1$. Finally, when $\pdf\in\mathcal D_\infty$ and $f(x_i)<\infty$, then $\adown\in\mathcal{D}_0$. If we allow functions such that $f(x)\to\infty$ as $x\to x_i$, in this case $\adown$ can belong to any of these three sets, depending of the value of the transformation parameter $\alpha$.

\end{remark}
%
The next result studies the behavior of the down transformation with respect to rescalings of functions. To this end, for any $\kappa\in\Rset^+$, we define the transformed density
\begin{equation}\label{def:resc}
f^{[\kappa]}(x)=\kappa f(\kappa x).
\end{equation}

\begin{proposition}[Scaling changes and down transformation]\label{prop:scaling_down}
Let $\alpha\in\Rset\setminus\{2\}$ and $\kappa\in\Rset^+$. Then, for any probability density $f$ satisfying the requirements of Definition \ref{def:down}, we have
	\begin{equation}\label{eq:down_resc}
	\down{\alpha}[f^{[\kappa]}](s)=\left(\down{\alpha}[f](s)\right)^{[\kappa^{\alpha-2}]}.
	\end{equation}
When $\alpha=2$ one finds
	\begin{equation}\label{eq:down_resc_a2}
	\down{2}[f^{[\kappa]}(s)]=\down{2}[f](s+\ln \kappa).
	\end{equation}
\end{proposition}
\begin{proof}
On the one hand, after simple algebraic manipulations employing Definition~\ref{def:down} and Eq. \eqref{def:resc}, we find that for any $\alpha\in\Rset$
	\begin{equation}\label{eq:proof_scaling}
	\down{\alpha}[f^{[\kappa]}](s)=\kappa^{\alpha-2}\frac{f^\alpha (\kappa\,x(s))}{|f' (\kappa\,x(s))|}.
	\end{equation}
On the other hand, for $\alpha\neq 2$ the change of variable in Eq.~\eqref{eq:can_change} is given by
	$$
	s(x)=\frac{(f^{[\kappa]}(x))^{2-\alpha}}{\alpha-2}=\frac{\kappa^{2-\alpha} f (\kappa x)^{2-\alpha}}{\alpha-2},\quad\text{or equivalently,}\quad \kappa^{\alpha-2}s=\frac{ f (\kappa\, x(s))^{2-\alpha}}{\alpha-2}.
	$$
We derive from the above equation and Eq.~\eqref{eq:can_change} that
	\begin{equation}\label{eq:proof_scaling2}
	\kappa\, x(s)=x(\kappa^{\alpha-2}s).
	\end{equation}
We next substitute Eq.~\eqref{eq:proof_scaling2} in Eq.~\eqref{eq:proof_scaling} to obtain
	\begin{equation*}
	\down{\alpha}[f^{[\kappa]}](s)=\kappa^{\alpha-2}\frac{f^\alpha (x(\kappa^{\alpha-2}s))}{|f' (x(\kappa^{\alpha-2}s))|}=\kappa^{\alpha-2}\down{\alpha}[f](\kappa^{\alpha-2}s)=(\down{\alpha}[f](s))^{[\kappa^{\alpha-2}]},
	\end{equation*}
establishing thus Eq. \eqref{eq:down_resc} and thus completing the proof for $\alpha\in\Rset\setminus\{2\}.$  Fix now $\alpha=2$ and recall that, in this case, the change of variable in Eq.~\eqref{eq:can_change_a2} is given by
	\begin{equation*}
s(x)=-\ln(f^{[\kappa]}(x))=-\ln(\kappa\,f(\kappa\,x))=-\ln\kappa-\ln(f(\kappa\,x)),
	\end{equation*}
or equivalently,
	\begin{equation*}
s+\ln\kappa=-\ln(f(\kappa\,x(s))).
\end{equation*}
The latter gives
\begin{equation}\label{eq:proof_scaling3}
\kappa\,x(s)=x(s+\ln\kappa).
\end{equation}
Finally, we infer from Eq.~\eqref{eq:proof_scaling3} and Eq.~\eqref{eq:proof_scaling} that
$$
\down{2}[f^{[\kappa]}](s)=\frac{f^2 (\kappa\,x(s))}{|f' (\kappa\,x(s))|}=\frac{f^2 (x(s+\ln\kappa))}{|f' (x(s+\ln\kappa))|}=\down{2}[f^{[\kappa]}](s+\ln\kappa),
$$
which closes the proof.
\end{proof}
We end the section dedicated to the exploration of the down transformation by a classification of the tail as $s\to\infty$ of the down transformed density with respect to the parameter $\alpha$.
\begin{proposition}\label{prop:tail_down}[Behavior of the tail]
Let $f:[x_i,\infty)\rightarrow \Rset^+$ be a probability density as in Definition \ref{def:down} such that $f'(x)\sim-Cx^{-\eta-1}$ as $x\to\infty$ for some $C>0$ and $\eta>1$. Then for $\alpha>2$,
\begin{equation}\label{eq:tail_down}
\down{\alpha}[f](s)\sim K(C,\eta,\alpha)s^{-\widetilde \eta_\alpha},\quad \widetilde \eta_\alpha=\frac{\eta(\alpha-1)-1}{\eta(\alpha-2)},
\quad K(C,\eta,\alpha)=\frac{1}{\eta}\left(\frac{\eta}{C}\right)^{-\frac{1}{\eta}}(\alpha-2)^{\widetilde \eta_{\alpha}},
\end{equation}
as $s\to\infty$. When $\alpha=2$ we have
\begin{equation}\label{eq:tail_down_2}
\down{2}[f](s)\sim \frac{C}{\eta^2}e^{-\frac{s(\eta-1)}{\eta}}.
\end{equation}
When $\alpha<2$ the support becomes compact and the function behaves near the border (that is, $s\to 0^-$) as Eq.~\eqref{eq:tail_down}
\end{proposition}
\begin{proof}
By an easy consequence of the l'Hospital rule and the fact that $f\in L^1([x_i,\infty))$, we deduce that
$$
f(x)\sim\frac{C}{\eta}x^{-\eta}, \quad {\rm as} \ x\to\infty.
$$
We next derive from Definition ~\ref{def:down} that
\begin{equation}\label{proofeq:39}
\down{\alpha}[f](s)\sim \frac{C^{\alpha-1}}{\eta^{\alpha}}x(s)^{-\eta(\alpha-1)+1},
\end{equation}
while for any $\alpha\neq 2$, Eq.~\eqref{eq:can_change} gives
\begin{equation}\label{proofeq:40}
s(x)\sim\frac{C^{2-\alpha}x^{-\eta(2-\alpha)}}{(\alpha-2)\eta^{2-\alpha}}, \quad {\rm as} \ x\to\infty.
\end{equation}
On the one hand, for $\alpha>2$, we deduce from Eq. \eqref{proofeq:40} that
$$
x(s)\sim\left(\frac{\eta}{C}\right)^{-\frac{1}{\eta}}(\alpha-2)^{-\frac{1}{\eta(2-\alpha)}}s^{-\frac{1}{\eta(2-\alpha)}}, \quad {\rm as} \ s\to\infty,
$$
and thus simple algebraic manipulations lead to Eq.~\eqref{eq:tail_down}. On the other hand, for $\alpha<2$, we recall from Remark \ref{rem:can} that $s_f=0$ is the upper edge of the support of $\down{\alpha}[f]$ and Eq. \eqref{eq:tail_down} holds true as $s\to0$ by similar deductions as above. Finally, for $\alpha=2$, by inverting Eq.~\eqref{eq:can_change_a2} we get
$$
s(x)\sim-\ln(x^{-\eta})=\eta\ln\,x, \quad {\rm or \ equivalently}, \quad x(s)\sim e^{\frac{s}{\eta}},
$$
as $x\to\infty$, respectively $s\to\infty$. This equivalence, together with Eq.~\eqref{proofeq:39}, readily implies Eq.~\eqref{eq:tail_down_2}, completing the proof.
\end{proof}



\subsection{The up transformation}\label{sec:up}

We introduce now a second transformation that will be called \emph{the up transformation} throughout the paper. In contrast to the down transformation, this new transformation does not require any monotonicity of the probability density function.


\begin{definition}\label{def:up}
Let $\pdf: \supp\longrightarrow \mathbb R^+$ be a probability density function with $\supp=(x_i,x_f)$. For $\alpha\in\Rset\setminus\{2\}$ we introduce the up transformation $\up{\alpha}$ by
\begin{equation}\label{eq:up}
f^{\uparrow}_\alpha(u)=\up{\alpha}[\pdf(x)](u)=|(\alpha-2)x(u)|^\frac{1}{2-\alpha},\quad u'(x)= - |(\alpha-2)x|^\frac{1}{\alpha-2}f(x),
\end{equation}
while for $\alpha=2$ we set
\begin{equation}\label{eq:up2}
f^{\uparrow}_2(u)=\up{2}[\pdf(x)](u)=e^{-x(u)},\quad u'(x)=-e^xf(x),
\end{equation}
\end{definition}
We precise here that the definition of $u(x)$ in Eq. \eqref{eq:up} is taken up to a translation; that is, for simplicity, in the sequel we use the primitive
$$
u(x)=\int_x^{x_f}|(\alpha-2)x|^\frac{1}{\alpha-2}f(x)dx, \quad \alpha\neq2,
$$
and the similar one for $\alpha=2$, where $x_f\in\Rset\cup\{\infty\}$ is the upper edge of the support of the domain of $f$, whenever this integral is finite. In the contrary case, we can integrate from $x_i$ to $x$ or employing any intermediate point $x_0$ in $(x_i,x_f)$ if the integral is divergent at both ends.

\medskip

\noindent Observe that, as shown in Remark \ref{rem:can}, when applying the down transformation, $(\alpha-2)s(x)$ is always positive, while for the up transformation, this is not necessarily true for the analogous term $(\alpha-2)x(u)$, justifying the use of absolute values in Eq. \eqref{eq:up}. The next result proves that the up transformation is in fact is the inverse of the down transformation.
\begin{proposition}[Inversion]\label{prop:inv}
Let $\pdf$ a probability density and let $\adown=\down{\alpha}[f]$ and $\aup=\up{\alpha}[f]$ be its $\alpha-$order \textit{down /up transformations}. Then, up to a translation,
\begin{equation}
\pdf=\up{\alpha}[\adown]=\down{\alpha}[\aup].
\end{equation}
i.e., $\down{\alpha}\up{\alpha}=\up{\alpha}\down{\alpha}=\mathbb I,$ where $\mathbb I$ denotes the identity operator (up to a translation). Summarizing,
\begin{equation*}
\down{\alpha}^{-1}=\up{\alpha},\quad \up{\alpha}^{-1}=\down{\alpha}.
\end{equation*}
\end{proposition}	
\begin{proof}
Assume first that $\alpha\neq2.$ We infer from Eq.~\eqref{eq:can_change} that
$$
x(s)=f^{-1}\left(\left[(\alpha-2)s\right]^\frac1{2-\alpha}\right),
$$
recalling that $(\alpha-2) s$ is always a positive quantity. Then
\begin{equation}\label{eq:interm1}
\begin{split}
\down{\alpha}[\pdf](s) &=\pdf\left(f^{-1}\left(\left[(\alpha-2)s\right]^\frac1{2-\alpha}\right)\right)^\alpha \left|f'\left(f^{-1}\left(\left[(\alpha-2)s\right]^\frac1{2-\alpha}\right)\right)\right|^{-1}
\\
&=\left[(\alpha-2)s\right]^\frac \alpha{2-\alpha} \left|(f^{-1})'\left(\left[(\alpha-2)s\right]^\frac1{2-\alpha}\right)\right|.
\end{split}
\end{equation}
Letting $u=[(\alpha-2)s]^\frac{1}{2-\alpha}$, we deduce from Eq. \eqref{eq:interm1} that
\begin{equation*}
\left|(f^{-1})'(u)\right|=u^{-\alpha}\adown\left(\frac{u^{2-\alpha}}{\alpha-2}\right),
\end{equation*}
and by integrating the previous equality, we get
\begin{equation}\label{eq:interm2}
f^{-1}(u)=K+\int_u^{u_f} v^{-\alpha}\adown\left(\frac{v^{2-\alpha}}{\alpha-2}\right)dv=K+\int_{\frac{{u}^{2-\alpha}}{\alpha-2}}^{\frac{{u_f}^{2-\alpha}}{\alpha-2}} [(\alpha-2)s]^{\frac1{\alpha-2}}\adown\left(s\right)ds,
\end{equation}
We readily derive from Eq. \eqref{eq:interm2} that the probabilty density $\pdf(x)$ satisfies
\begin{equation}
x=K+\int_{\frac{{f(x)}^{2-\alpha}}{\alpha-2}}^{\frac{{u_f}^{2-\alpha}}{\alpha-2}} [(\alpha-2)s]^{\frac1{\alpha-2}}\adown\left(s\right)ds.
\end{equation}
We have thus shown that $\up{\alpha}\down{\alpha}=\mathbb I,$ for any $\alpha\neq 2.$

\smallskip

Conversely, let us start from $\aup(u)$ and compute $\down{\alpha}[\aup]$. We have
\[\frac{d\aup}{du}=-[(\alpha-2)s(u)]^{\frac{\alpha-1}{2-\alpha}} s'(u)=-\frac{[(\alpha-2)s(u)]^\frac{\alpha}{2-\alpha}}{f(s(u))}, \]
from where
\[\frac{d\aup}{du}=-\frac{[\aup(u)]^\alpha}{f(s(u))},\]
hence
\begin{equation}\label{eq:interm3}
f(s)=\frac{[\aup(u(s))]^\alpha}{|\frac{d\aup}{du}(u(s))|}.
\end{equation}
Furthermore
\[u'(s)=[(\alpha-2)s]^{\frac1{\alpha-2}}f(s),\]
from where
\begin{equation}\label{eq:interm4}
s'(u)=\frac1{[(\alpha-2)s(u)]^{\frac1{\alpha-2}}f(s(u))}=\frac{\aup(u)}{f(s(u))}=\aup(u)^{1-\alpha}|(\aup)'(u)|.
\end{equation}
Note that Eqs. \eqref{eq:interm3} and \eqref{eq:interm4} ensure that $\down{\alpha}\up{\alpha}=\mathbb I$, completing the proof when $\alpha\neq 2.$

\smallskip
\noindent
Fix now $\alpha=2$ for the rest of the proof. From Eq.~\eqref{eq:can_change_a2} follows
\begin{equation}\label{eq:proof_a2(1)}
f(x(s))=e^{-s},\quad\text{that is},\quad x(s)=f^{-1}\left(e^{-s}\right),
\end{equation}
thus, applying the definition ~\eqref{def:down} for $\alpha=2$ we get
\begin{equation}\label{eq:proof_a2(2)}
\down{2}[f](s)=\frac{f\left(f^{-1}\left(e^{-s}\right)\right)^2}{f'\left(f^{-1}\left(e^{-s}\right)\right)}=e^{-2s}\left(f^{-1}\right)' \left(e^{-s}\right).
\end{equation}
We perform the change of variable $u=e^{-s}$ in Eq.~\eqref{eq:proof_a2(2)} to deduce
\begin{equation*}
u^{-2}\down{2}[f](-\ln u)=\left(f^{-1}\right)' \left(u\right)
\end{equation*}
and by integration we obtain (up to a translation)
\begin{equation*}
f^{-1}(u)=\int_{u_f}^{u} v^{-2}\down{2}[f](-\ln v)dv=-\int_{-\ln u_f}^{-\ln u} e^w \down{2}[f](w)dw=\int_{s}^{s_f}e^w\down{2}[f](w)dw,
\end{equation*}
which, together with Eq.~\eqref{eq:proof_a2(1)}, implies that $\up{2}\down{2}=\mathbb I.$

\smallskip
\noindent
We next prove the converse implication. Let us start by computing the derivative of the up transformation $\up{2}[f],$ which, together with Eq.~\eqref{eq:up2}, gives
\begin{equation}\label{eq:proof_a2(3)}
\frac{df^{\uparrow}_2(u)}{du}=-e^{-s(u)}s'(u)=\frac{1}{e^{2s(u)}f(s(u))}.
\end{equation}
We infer from the definition of the down transformation and Eq.~\eqref{eq:proof_a2(3)} that
\begin{equation}
\down{2}[f^{\uparrow}_2](s)=\frac{(f^{\uparrow}_2(u(s)))^2}{(f^{\uparrow}_2)'(u(s))}=e^{-2s}e^{2s}f(s)=f(s),
\end{equation}
which concludes the proof.
\end{proof}
%


%
\begin{remark}
From Definition~\ref{def:up} we infer that the length of the support of $\up{\alpha}[f]$, for $\alpha\neq2$, can be finite or infinite and is given by $\int_{x_i}^{x_f} |(\alpha-2)v|^\frac1{\alpha-2}f(v)dv$, independent of the translation and of the convergence or not of this integral. The same occurs for $\alpha=2$.
\end{remark}

We describe in the next two results, similarly as at the end of Section \ref{sec:down}, the behavior of the up transformation with respect to the scaling Eq. \eqref{def:resc} and when applied to probability density function $f$ presenting a suitable algebraic tail as $x\to\infty$. Due to technical difficulties, we have to leave out of these statements the limiting case $\alpha=2$, which will be considered separately.

\begin{proposition}\label{prop:scaling_up}[Scaling changes and up transformation] Let $\alpha\in\Rset\setminus\{2\}$ and $\kappa\in\Rset^+$. Then, for any probability density function $f$, we have
		\begin{equation}\label{eq:scaling_up}
		\up{\alpha}[f^{[\kappa]}](u)=\left(\up{\alpha}[f](u)\right)^{[\kappa^{1/(\alpha-2)}]},
		\end{equation}
recalling that $f^{[\kappa]}$ has been defined in Eq. \eqref{def:resc}.
\end{proposition}	
\begin{proof}
We apply the down transformation to the right hand side of Eq. \eqref{eq:scaling_up} and we obtain from Proposition \ref{prop:scaling_down} and Proposition \ref{prop:inv} that
$$
\down{\alpha}\left[\left(\up{\alpha}[f](u)\right)^{[\kappa^{1/(\alpha-2)}]}\right](s)
=\left(\down{\alpha}\left[\up{\alpha}[f]\right](s)\right)^{[\kappa]}=f^{[\kappa]}(x)=\down{\alpha}\left[\up{\alpha}[f^{[\kappa]}](u)\right](s).
$$
The injectivity of the down transformation then leads to Eq. \eqref{eq:scaling_up}, completing the proof.
\end{proof}
\begin{proposition}\label{prop:tail_up}[The behavior of the tail of the up transformed density]
Let $f:[x_i,\infty)\rightarrow \Rset^+$ be a probability density function such that $f(x)\sim Cx^{-\eta},$ as $x\to\infty$, for some $\eta>1$ and $C>0$. Set $\alpha_c=2+\frac{1}{\eta-1}>2$. Then, for any $\alpha\in(2,\alpha_c)$, the transformed density $\up{\alpha}[f]$ has an unbounded support and the tail decays as
\begin{equation}\label{propeq:tail_up}
	\up{\alpha}[f](u)\sim\left|K_1 u\right|^{\frac 1{\eta(\alpha-2)+1-\alpha}}, \quad {\rm as} \ u\to-\infty,
\end{equation}
for some $K_1>0$. When $\alpha=\alpha_c>2$, $\up{\alpha_c}[f]$ has an exponential decay, while for either $\alpha<2$ or $\alpha>\alpha_c$, the support becomes compact. Finally, $\up{\alpha}[f]$ tends to the uniform density in the limit $\alpha\to\infty$.
\end{proposition}

\begin{proof}
Let us fix $\alpha>2$. We infer from the definition of the up transformation that
\begin{equation}\label{eq:interm7}
\up{\alpha}[f](u)=|(\alpha-2)x(u)|^{\frac{1}{2-\alpha}},
\end{equation}
where
\begin{equation}\label{eq:interm5}
\begin{split}
u(x)&=-\int_{x_i}^{x}|(\alpha-2)v|^{\frac{1}{\alpha-2}}f(v)\,dv\simeq-\int_{x_i}^{x}|(\alpha-2)v|^{\frac{1}{\alpha-2}}Cv^{-\eta}\,dv\\
&=-C(\alpha-2)^{\frac{1}{\alpha-2}}\int_{x_i}^{x}v^{\frac{1}{\alpha-2}-\eta}\,dv.
\end{split}
\end{equation}
We now split the analysis into three cases:

$\bullet$ if $\alpha\in(2,\alpha_c)$, we deduce that
$$
\frac{1}{\alpha-2}-\eta\in(-1,\infty),
$$
thus the improper integral in Eq. \eqref{eq:interm5} is infinite and thus the support of $\up{\alpha}[f]$ is unbounded. Moreover, we further infer from Eq. \eqref{eq:interm5} that
\begin{equation}\label{eq:interm6}
u(x)\sim-\frac{C(\alpha-2)^{\frac{\alpha-1}{\alpha-2}}}{\alpha-1-\eta(\alpha-2)}x^{\frac{1}{\alpha-2}-\eta+1}, \quad {\rm as} \ x\to\infty,
\end{equation}
noticing that
$$
\alpha-1-\eta(\alpha-2)>0, \quad {\rm since} \quad 2<\alpha<\alpha_c.
$$
By inverting the equivalence Eq. \eqref{eq:interm6} and thus expressing $x(u)$ in terms of $u$, we are led to the tail of $\up{\alpha}[f]$ as in Eq. \eqref{propeq:tail_up}, with some constant $K_1>0$ that can be made explicit in terms of $\alpha$, $C$ and $\eta$ (we omit here the precise formula for simplicity).

$\bullet$ if $\alpha=\alpha_c$, then $\frac{1}{\alpha-2}-\eta=-1$ and we obtain from Eq. \eqref{eq:interm5} that
$$
u(x)=-C(\eta-1)^{1-\eta}\ln\,\frac{x_i}{x}, \quad {\rm or \ equivalently}, \quad x(u)=x_ie^{-\frac{|u|(\eta-1)^{\eta-1}}{C}},
$$
and we obtain the claimed exponential tail from Eq. \eqref{eq:interm7}.

$\bullet$ if $\alpha\in(\alpha_c,\infty)$, the function $u(x)$ given by the integral in Eq. \eqref{eq:interm5} is now bounded, and thus $\up{\alpha}[f]$ is compactly supported. A similar property holds true also for $\alpha<2$, since $\frac{1}{\alpha-2}-\eta<-\eta<-1$ and thus $u(x)$ defined in Eq. \eqref{eq:interm5} remains bounded. Finally, the uniform (constant) density on a bounded support is easily obtained from Eq. \eqref{eq:interm7} by taking limits as $\alpha\to\infty$.
\end{proof}
We refrain from studying exact expressions of the tail behavior in the case $\alpha=2$ in Proposition \ref{prop:tail_up}, since it leads to very complex expressions in terms of special functions (if any). However, in the next proposition we give some results concerning the tail of $\up2[f]$ for some particular forms of the density $f$.
\begin{proposition}\label{prop:up2_tail}
Let $f_j:[x_i,\infty)\rightarrow \Rset^+$, $j=1,2,3$, be probability density functions such that
$$
f_1(x)\sim C_1e^{-x},\qquad f_2(x)\sim C_2 e^{-\xi x},\qquad f_3(x)\sim C_3x^{-\eta}
$$
as $x\to\infty$, for some $\eta>1,\,\xi\in(0,1)$ and $C_j>0$. Then, the transformed density $\up{2}[f_j]$ has an unbounded support and the following statements hold true
\begin{enumerate}
\item $\up{2}[f_1]$  decay exponentially and, consequently, $\sigma_p[\up{2}[f_1]]<\infty$ for any $p>0.$
\item $\up{2}[f_2]$  decay as $[(1-\xi)|u|+k]^{\frac1{\xi-1}}$ and then $\sigma_p[\up{2}[f_2]]<\infty$ for any $p\in\left(0,\frac{\xi}{1-\xi}\right).$
\item $\up{2}[f_3]$ has no finite moments; that is, $\sigma_p[\up{2}[f_3]]=\infty$ for any $p>0.$
\end{enumerate}
\end{proposition}
\begin{proof}
We give a formal proof, assuming that $x$ is taken sufficiently large in the lines below. Starting from $f_1(x)\simeq e^{-x},$ when $\alpha=2$ one finds
$$
\up{2}[f_1]=e^{-x(u)}, \quad u(x)\simeq -\int_{x_i}^{x} e^ve^{-v} dv=x_i-x
$$
from where $\up 2[e^{-x}]\simeq \overline{C}_1e^{u}=\overline{C}_1e^{-|u|}.$

For $f_2\simeq e^{-\xi x}$, $\xi\in(0,1)$, one finds
$$
u(x)\simeq -\int_{x_i}^{x} e^ve^{-\xi v} dv=-\frac{e^{(1-\xi)x}-k}{1-\xi},
$$
that is,
$$
x(u)\simeq
\frac{ \log[(\xi-1)u+k]}{1-\xi},
$$
hence
$$\up 2[e^{-\xi x}]=e^{-x(u)}\simeq[(\xi-1)u+k]^{\frac1{\xi-1}}=\left[(1-\xi)|u|+k\right]^{\frac1{\xi-1}}.$$

Finally, starting from $f_3(x)\simeq x^{-\eta}$, when $\alpha=2$ the change of variable is given by
$$
u(x)\simeq -\int_{x_i}^{x} e^vv^{-\eta } dv,
$$
and then
$$
\int_\Rset |u|^p\; \up{2}[f_3](u) du = \int_\Rset \left|\int_{x_i}^{x} e^vf_3(v)dv\right|^pf_3(x) dx\sim\int_\Rset\left|\int_{x_i}^{x} e^vv^{-\eta } dv\right|^p x^{-\eta} dx.
$$
Now we observe that, for any $p>0$ and $\eta>1$, for $M>0$ sufficiently large we have
\begin{equation*}
\int_{\Rset}\left|\int_{M}^{x} e^vv^{-\eta } dv\right|^p x^{-\eta} dx \geqslant \int_{\Rset}\left|\int_{M}^{x} e^v dv\right|^p x^{-\eta p -\eta} dx  \geqslant \int_{M}^{\infty}\left|e^{x}-e^{M}\right|^p x^{-\eta p -\eta} dx=\infty.
\end{equation*}
A rigorous proof follows along the same lines by splitting the integrals in all the steps above into integrals from $x_i$ to $M$ and from $M$ to $x$ for $M$ sufficiently large.
\end{proof}

\begin{remark}
Note that the up transformation $\up{\alpha}$ with $\alpha\neq 2$ commutes with the scaling, as stated in the in Eq.~\eqref{eq:scaling_up} in Proposition~\ref{prop:scaling_up}. In stark contrast, this fact is not true when $\alpha=2$. Actually, as highlighted in Proposition~\ref{prop:up2_tail}, the effect in the tail behavior of the transform $\up 2$ on densities $f_1$ and $f_2$ is completely different, while $f_1$ and $f_2$ are equal up to a scaling change. More precisely, $\up2[f_2]$ is a heavy-tailed density but $\up2[f_1]$ has exponential decay.
\end{remark}


\section{Up/down transformations and informational measures}\label{sec:measures}

In this section we explore, in a first step, the behavior of the main informational measures of the transformed up/down probability density functions with respect to their non-transformed counterparts. In a second step, we apply these results to connect the Hausdorff moment problem with its entropic (see Romera et al.~\cite{Romera2001}) and Fisher counterparts. Furthermore, some implications in the framework of the MaxEnt approach are investigated.

\subsection{Moments, Shannon and Rényi entropy and Fisher information}

We start with some simple relations between the power R\'enyi entropy, the $p$-typical deviation and the $(p,\lambda)-$Fisher information, that we gather in the next result.

\begin{lemma}\label{lem:MEF}
Let $\pdf$ be a probability density and $\aup$ and $\adown$ its up/down transformations. Then, if $\alpha\in\Rset\setminus\{2\}$, the following equalities hold true:
\begin{equation}\label{eq:ME}
\sigma_p[\adown]=\frac{N_{1+(2-\alpha)p}^{\alpha-2}[\pdf]}{|2-\alpha|},\quad\text{or equivalently,}\quad N_{\lambda}[\aup]=\left(|2-\alpha| \sigma_{\frac{\lambda-1}{2-\alpha}}[\pdf]\right)^{\frac{1}{\alpha-2}},
\end{equation}
and
\begin{equation}\label{eq:EF}
N_{\lambda}[\adown]=\phi_{1-\lambda,2-\alpha}^{2-\alpha}[\pdf],\quad\text{or equivalently,}\quad  \phi_{p,\beta}[f^{\uparrow}_{2-\beta}]=\left(N_{1-p}[\pdf]\right)^{\frac1{\beta}}.
\end{equation}
Finally, for the Shannon entropy we have
\begin{equation}\label{eq:SUD}
S[\adown]=\alpha S[f]+\big \langle \log \left| f'\right|\big \rangle,\qquad S[\aup]=\frac1{2-\alpha}\left\langle \log |x|\right \rangle+\frac{\log|2-\alpha|}{2-\alpha}.
\end{equation}
\end{lemma}

Note that the left equality in Eq. \eqref{eq:ME} can be seen as a particular case of \cite[Theorem 2]{Mnatsakanov2008}, but all the other equalities, to the best of our knowledge, are completely new. We give below a full proof for the sake of completeness.
\begin{proof}
We proceed by direct calculation employing Definition \ref{def:down}, taking into account that $\alpha\neq2$. We have
\begin{equation*}
\begin{split}
\sigma_p[\down{\alpha}[\pdf]]&=\left[\int_{\Rset}|s|^p\down{\alpha}[\pdf](s)ds\right]^{\frac{1}{p}}
=\left[\int_{\Rset}|s|^p\frac{f^{\alpha}(x(s))}{|f'(x(s))|}ds\right]^{\frac{1}{p}}\\
&=\frac1{|2-\alpha|}\left[\int_{\mathbb R^+}\pdf(x)^{1+(2-\alpha)p}dx\right]^{\frac{1}{p}}=\frac{1}{|2-\alpha|}N_{1+(2-\alpha)p}^{\alpha-2}[\pdf],
\end{split}
\end{equation*}
from which the first equality in Eq. \eqref{eq:ME} follows. The second equality in Eq. \eqref{eq:ME} is immediately obtained from the previous one by setting $\lambda=1+(2-\alpha)p$ and taking into account that the down and up transformations are mutually inverse, as established in Proposition \ref{prop:inv}. Observe also that $\alpha\neq2$ is equivalent to $\lambda=1+(2-\alpha)p\neq1$, thus the limiting case $\lambda=1$ cannot be attained in \eqref{eq:ME}.

\smallskip

\noindent For the second equality, we have
\begin{equation*}
\begin{split}
N_\lambda[\down{\alpha}[\pdf]]&=\left[\int_{\Rset}[\down{\alpha}[\pdf](s)]^\lambda ds\right]^{\frac{1}{1-\lambda}}=\left[\int_{\Rset}\frac{f^{\alpha\lambda}(x(s))}{|f'(x(s))|^{\lambda}}ds\right]^{\frac{1}{1-\lambda}}\\
&=\left[\int_{\Rset}\pdf^{1+(\lambda-1)\alpha}|f'(x)|^{1-\lambda}dx\right]^{\frac{1}{1-\lambda}}=\phi_{1-\lambda,2-\alpha}^{2-\alpha}[\pdf],
\end{split}
\end{equation*}
which gives the first equality in Eq. \eqref{eq:EF}. The second equality in Eq. \eqref{eq:EF} follows immediately by setting $p=1-\lambda$, $\beta=2-\alpha$ and recalling Proposition \ref{prop:inv}. Notice once again that $\beta\neq0$ since $\alpha\neq2$, thus the Fisher information is well defined.

\smallskip

\noindent In order to establish Eq.~\eqref{eq:SUD}, we proceed by direct substitution of the down and up transformed densities as given in Definitions \ref{def:down} and ~\ref{def:up} in the definition of the Shannon entropy. We have, for the down transformation,

\begin{equation*}
\begin{split}
S[\adown]&=-\int_{\Rset}\adown(s)\log\,\adown(s)\,ds=-\int_{\Rset}\frac{f^{\alpha}(x)}{|f'(x)|}\log\frac{f^{\alpha}(x)}{|f'(x)|}f^{1-\alpha}(x)|f'(x)|\,dx\\
&=-\int_{\Rset}f(x)(\alpha\log\,f(x)-\log|f'(x)|)\,dx=\alpha S[f]+\big \langle \log \left| f'\right|\big \rangle,
\end{split}
\end{equation*}
respectively, for the up transformation,
\begin{equation*}
\begin{split}
S[\aup]&=-\int_{\Rset}\aup(u)\log\aup(u)\,du=\int_{\Rset}|(\alpha-2)x|^{\frac{1}{2-\alpha}}\log\,|(\alpha-2)x|^{\frac{1}{2-\alpha}}|(\alpha-2)x|^{\frac{1}{\alpha-2}}f(x)\,dx\\
&=\int_{\Rset}\frac{f(x)}{2-\alpha}\left(\log|x|+\log|2-\alpha|\right)\,dx=\frac1{2-\alpha}\left\langle \log |x|\right \rangle+\frac{\log|2-\alpha|}{2-\alpha},
\end{split}
\end{equation*}
and the proof is complete.
\end{proof}
\smallskip

\noindent \textbf{The exponent $\alpha=2$.} If we let $\alpha=2$, the change of variable in the previous calculations switches to $s(x)=-\ln\,f(x)$, according to Eq. \eqref{eq:can_change_a2}. Thus, we no longer obtain Eq. \eqref{eq:ME} and Eq. \eqref{eq:EF}. Indeed, direct calculations similar to the ones in the proof of Lemma \ref{lem:MEF} give
\begin{equation}\label{eq:Sp_down2}
\sigma_p[\down{2}[\pdf]]=\left[\int_{\Rset}f(x)|\ln\,f(x)|^pdx\right]^{\frac{1}{p}}=\overline S_p[f],
\end{equation}
where $\overline S_p[f]$ is defined in Eq.~\eqref{eq:Sp}. Moreover,
\begin{equation}\label{eq:Ren_down2}
N_\lambda[\down{2}[\pdf]]=\left[\int_{\Rset}|f^{-2}(x)f'(x)|^{1-\lambda}f(x)\,dx\right]^{\frac{1}{1-\lambda}}=\lim_{\widetilde \lambda\to0}\phi_{1-\lambda,\widetilde \lambda}[f]^{\widetilde\lambda}\equiv F_{1-\lambda,0}[f]^\frac{1}{1-\lambda}.
\end{equation}
For the up transformed density with $\alpha=2$ one finds
\begin{equation}\label{eq:N_up2}
N_{\lambda}[f^{\uparrow}_2]=\left(\int_{\Rset} e^{-(\lambda-1)x} f(x)dx\right)^{\frac{1}{1-\lambda}}=\left\langle e^{-(\lambda-1) x}\right \rangle^{\frac 1{1-\lambda}}_f=(\sigma_{\lambda-1}^{(E)}[f])^{-1}
\end{equation}
where $\sigma_{\lambda-1}^{(E)}$ is defined in Eq.~\eqref{eq:sigmaE}. For the Shannon entropy, we have
\begin{equation}\label{eq:Shanup2}
S[f^{\uparrow}_2]=-\int_{\Rset}f^{\uparrow}_2(u)\ln f^{\uparrow}_2(u)\,du=-\int_{\Rset}f(x)\ln e^{-x}\,dx=\int_{\Rset}f(x)x\,dx=\langle x \rangle_f
\end{equation}
Note that Eq. \eqref{eq:Shanup2} together with the inversion Proposition~\ref{prop:inv} give that
	$$
	S[f]=\langle s\rangle_{f^{\downarrow}_{2}}
	$$
	where $s$ is the variable of the down transform $f^{\downarrow}_{2}.$

\smallskip

Our next goal is to study the effect of the composition of up and down transformations with different exponents. We thus compute in the next result the $p$-typical deviation of a composed application of an up and then down transformation with different parameters to a probability density function $f$. To this end, for a probability density function $f$ we introduce the notation
$$
f^{\uparrow\downarrow}_{\alpha\beta}=\down{\beta}\left[\up{\alpha}[f]\right].
$$
\begin{proposition}[Composed up and down transformations]\label{prop:momentum-shift}
If $\alpha,\beta\in\Rset\setminus\{2\}$ we have
\begin{equation}\label{eq:momentum-shift}
\sigma_{p}[f^{\uparrow\downarrow}_{\alpha\beta}]=\frac{|\alpha-2|^{\frac{\beta-2}{\alpha-2}}}{|\beta-2|}\sigma_{\frac{\beta-2}{\alpha-2}p}^{\frac{\beta-2}{\alpha-2}}[f]
\end{equation}
For $\alpha=2$ and $\beta\neq2$ we have
\begin{equation}\label{eq:momentum-shift(2)}
\sigma_{p}[f^{\uparrow\downarrow}_{2 \beta}]=\frac1{|2-\beta|}(\sigma^{(E)}_{(2-\beta)p})^{2-\beta}.
\end{equation}
Finally, for $\alpha\neq2$ and $\beta=2$ we have
\begin{equation}\label{eq:momentum-shift(3)}
\sigma_{p}[f^{\uparrow\downarrow}_{\alpha 2}]=\frac1{|\alpha-2|}\int_{\Rset} f(x)|\log |(\alpha-2)x||^pdx.
\end{equation}
Note that in the case $\alpha=1$ or $\alpha=3$ we have $\sigma_{p}[f^{\uparrow\downarrow}_{\alpha 2}]=\sigma_{p}^{(L)}[f],$ where $\sigma_p^{(L)}$ is defined in Eq.~\eqref{eq:sigmaL}.
\end{proposition}
\begin{proof}
Let $\alpha,\beta\in\Rset\setminus\{2\}$. We obtain as an easy consequence of Eq. \eqref{eq:ME} that
\begin{equation*}
\begin{split}
\sigma_p[f^{\uparrow\downarrow}_{\alpha\beta}]&=\frac{N^{\beta-2}_{1+(2-\beta)p}[\up{\alpha}[f]]}{|2-\beta|}
=\frac{1}{|2-\beta|}\left[|2-\alpha|\sigma_{\frac{(2-\beta)p}{2-\alpha}}\right]^{\frac{\beta-2}{\alpha-2}}\\
&=\frac{1}{|2-\beta|}|2-\alpha|^{\frac{\beta-2}{\alpha-2}}\sigma_{\frac{(2-\beta)p}{2-\alpha}}^{\frac{\beta-2}{\alpha-2}},
\end{split}
\end{equation*}
as stated in Eq. \eqref{eq:momentum-shift}. In the case $\alpha=2$, $\beta\in\Rset\setminus\{2\}$, Eq. \eqref{eq:momentum-shift(2)} stems from an application of Eq. \eqref{eq:N_up2},
$$
\sigma_{p}[f^{\uparrow\downarrow}_{2 \beta}]=\sigma_p[\down{\beta}[f^{\uparrow}_2]]=\frac{N^{\beta-2}_{1+(2-\beta)p}[f^{\uparrow}_2]}{|2-\beta|}
=\frac1{|2-\beta|}
(\sigma^{E}_{(2-\beta)p})^{\beta-2},
$$
while for $\alpha\in\Rset\setminus\{2\}$, Eq. \eqref{eq:momentum-shift(3)} follows from the following calculations employing Eq. \eqref{eq:Sp_down2}
$$
\sigma_{p}[f^{\uparrow\downarrow}_{\alpha 2}]=\sigma_p[\down{2}[f^{\uparrow}_\alpha]]=S_p[f^{\uparrow}_\alpha]
=\frac1{|\alpha-2|}\int_{\Rset} f(x)|\log |(\alpha-2)x||^pdx,
$$
completing the proof.
\end{proof}

\begin{remark}
A direct calculation shows that the moments of the simple transformation
$$
\widetilde f_{\alpha \beta}:=\frac{\alpha-2}{\beta-2}x^{\frac{\alpha-\beta}{\beta-2}}f(x^{\frac{\alpha-2}{\beta-2}})
$$
satisfy $\sigma_p[\widetilde f_{\alpha \beta}]=\sigma_p[f^{\uparrow\downarrow}_{\alpha\beta}]$. Then, if we assume that the set $\sigma_p[\pdf]$ uniquely characterizes the probability density $f$, we deduce that for any pair of numbers $\alpha$ and $\beta$ in $\Rset\setminus\{2\}$, we have
$$
f^{\uparrow\downarrow}_{\alpha\beta}(x)=\frac{\alpha-2}{\beta-2}x^{\frac{\alpha-\beta}{\beta-2}}f(x^{\frac{\alpha-2}{\beta-2}}).
$$
\end{remark}

\subsection{Dealing with the Hausdorff entropic moment problem}
	
A set of numbers $\{m_i\}_{i=1}^\infty$ is said to be a \emph{moment-sequence} when there exists a function (a probability density in the context of this paper) such that $m_i=\widetilde\mu_i[f],$  for $i=1,2,\ldots,$ where
$$
\widetilde \mu_i[f]=\int_\Rset x^i f(x) dx =\langle x^i\rangle_f,\qquad i\in\mathbb N.
$$
The Hamburger, Stieltjes and Hausdorff moment problems consist in determining the conditions under which a sequence of numbers $\{m_i\}_{i=1}^\infty$ uniquely characterizes a probability density, when the respective supports are $\Rset,\Rset^+$ and the interval $(0,1)$ and, in such case, find the corresponding probability density. Although reconstruction methods have been investigated by different formal approaches (see~\cite{Shohat1970, Akhiezer_Moment_Problem}), the numerical reconstruction of a probability density employing a large number of moments is unstable in the Hadamard sense. This fact has motivated a quest for establishing different approaches to tackle with this problem, e.g. formal reconstructions methods by using Pollaczek Polynomials~\cite{Viano1991} or Chrystoffel functions~\cite{Gavriliadis2008,Gavriliadis2009} have been proposed. In an analogous way, the entropic Hausdorff problem consists in determining the conditions under which a set of numbers corresponds with the family of Rényi entropies (or equivalently $L^p$-norms) of a certain decreasing probability density function~\cite{Romera2001}. In the following lines we show that the up/down transformations introduced in this paper give a direct connection between the standard and the entropic moment problems. In fact, similar transformations to the down and up transforms are implicitly used in the proof of the results in~\cite{Romera2001}. We are aware of the fact that the optimal conditions for the uniqueness in the Hausdorff moment problem are still an open problem; however, in the next lines, we begin by assuming that this uniqueness is fulfilled, in order to relate it to other reconstruction problems.

Given $\alpha\neq2$ and a set of values $\{m_i\}_{i=1}^\infty$ which determines certain probability density $f$, the same set of values also determines the rearrangement of another probability density with Rényi entropy power $\{N_\lambda\}_{\lambda\in A},$ where
$$
A=\{1+i(2-\alpha):\;i\in\mathbb N\},\quad N_{1+i(2-\alpha)}=\left(|2-\alpha| m_{i}^\frac 1i[\pdf]\right)^{\frac{1}{\alpha-2}}.
$$
This follows from the fact that the $p$-moment of a probability density essentially corresponds with the power Rényi entropy (or the $L^p$ norm) of the down transformed density, as established in Lemma~\ref{lem:MEF}, together with the invertibility of the down transformation proved in Proposition~\ref{prop:inv}.	

\begin{definition}
We say that a set of numbers $\{N_\lambda\}_{\lambda\in A},$ with $card(A)=card(\mathbb N),$ is a \emph{entropy-sequence} if there is a probability density $f$ such that $N_\lambda=N_\lambda[f],$ for any $\lambda\in A$.
\end{definition}

With this definition, we are in a position to give a precise statement of our reconstruction theorem.
\begin{theorem}Let $\alpha\in(-\infty,2)$ and a set of real values $\{N_\lambda\}_{\lambda\in A}^\infty$ with $A=\{1+i(2-\alpha):\;i\in\mathbb N\}$. If the set of values $\{m_i\}_{i=1}^\infty$ such that
\begin{equation}\label{eqproof:78}
m_i=\frac{1}{|\alpha-2|^i}N_{1+(2-\alpha)i}^{(\alpha-2)i}
\end{equation}
is a moment-sequence of a probability density $f$ (i.e., $\mu_i[f]=m_i$), then the set $\{N_\lambda\}_{\lambda\in A}^\infty$ is an entropy-sequence that uniquely characterizes the decreasing rearragement $g^*$ of any probability density $g$ with power Rényi entropies $N_\lambda.$ In particular, $g^*=\up{\alpha}[f]$.
\end{theorem}
\begin{proof}
Let $\{N_\lambda,\lambda=1+(2-\alpha)i\}_{i=1}^\infty$ be a sequence of numbers such that $m_i$ defined by~\eqref{eqproof:78} is a moment-sequence which uniquely characterizes a probability density $f$; that is, $\mu_i[f]=m_i.$ Introduce now $g^*=\up{\alpha}[f]$. Then, using Lemma~\ref{lem:MEF}, we have that
\begin{equation}
\mu_{i}[\down{\alpha}[g^{*}]]=\frac{1}{|\alpha-2|^i}N_{1+(2-\alpha)i}^{(\alpha-2)i}[g]=m_i.
\end{equation}
Since the set $\{m_i\}_{i=1}^\infty$ determines uniquely the probability density $f$, we deduce that $f=\down{\alpha}[g^*]$. Moreover, if $g$ is a probability density function such that $N_\lambda[g]=N_\lambda=N_\lambda[g^*]$, then the sequence $N_\lambda$ uniquely reconstructs $g^*$, concluding the proof.
\end{proof}

As previously mentioned, when the number of fixed moments is large, the reconstruction problem becomes numerically unstable. The MaxEnt approach has been widely used in order to avoid this problem. It consists in maximizing an entropic functional employing a smaller number of fixed moments. An adaptation of the MaxEnt approach in the framework of the Hausdorff entropic problem is used in \cite{Romera2001}; more precisely, they reconstruct a decreasing probability density by fixing a lower number $N$ of Rényi entropies and maximizing the Fisher information \eqref{eq:def_FI}. In the next theorem, which is now an immediate consequence of the analysis performed in the previous sections, we relate the probability density maximizing Shannon and Rényi entropies with a set of fixed moments with the probability density maximizing Fisher-like functionals when a set of Rényi entropies is fixed.

\begin{theorem}
Let $\{m_i\}_{i=1}^N$ be a set of real numbers and let $\mathcal D$ be the set of probability density functions $f$ such that $\sigma_i[f]=m_i$. For some $\alpha<2$ let $\overline \D_\alpha$ be the set of probability densities $g$ such that $N_{1+(2-\alpha) i}[g]=((2-\alpha) m_i)^\frac1{\alpha-2}$. Let $\displaystyle f_\lambda=\argmax_{f\in\D}[R_\lambda]$, for $\lambda\in\Rset$. Then, for $\lambda\neq1$ we have
$$
\displaystyle \up{\alpha}[f_\lambda]=\argmax_{g\in\overline\D_\alpha}[\phi_{1-\lambda,2-\alpha}[g]],
$$
while for $\lambda=1$ we have
$$
\displaystyle \up{\alpha}[f_1]=\argmax_{g\in\overline\D_\alpha}[S[g]+\big\langle\log (g') \big\rangle].
$$
\end{theorem}
\begin{proof}
The proof follows directly from Lemma~\ref{lem:MEF} and the inversion property in Proposition~\ref{prop:inv}.
\end{proof}

In the following lines we explore the use of up and down transformations to investigate a Fisher-like moment problem. For a probability density function $f$ allowing for an application of the down transformation twice, we introduce the following notation:
\begin{equation*}\label{prop:downdown}
f^{\downarrow\downarrow}_{\alpha\beta}=\down{\beta}\left[\down{\alpha}[f]\right].
\end{equation*}
We compute in the following statement the $p$-th moments of the density $f^{\downarrow\downarrow}_{\alpha\beta}$.
\begin{proposition}
Let $f$ be a strictly decreasing probability density and let $\alpha,\beta$ be two real numbers such that
$$
\max\left(\frac{\pdf(x)\pdf''(x)}{[\pdf'(x)]^{2}}\right)<\alpha<2 \quad {\rm and} \quad \beta<2.
$$
Then
\begin{equation}\label{eqprop:downdown}
\mu_p[f^{\downarrow\downarrow}_{\alpha\beta}]=\frac{\phi_{r,\xi}^{r\xi}[f]}{|\alpha-2|^{p}},
\qquad r=p(\beta-2),\quad \xi=2-\alpha.
\end{equation}
\end{proposition}
\begin{proof}
Given a strictly decreasing probability density $f$ and $\alpha$ a real number such that $\alpha>\max\left(\frac{\pdf(x)\pdf''(x)}{[\pdf'(x)]^{2}}\right),$ the density $\down{\beta}\left[\down{\alpha}[f]\right]$ is well-defined for any $\beta\in \Rset$ (as discussed in Remark~\ref{rem:downcomposition}). The restriction $\alpha,\beta<2$ is imposed to keep the support compact. Then, it follows from Lemma~\ref{lem:MEF} that
\begin{equation}
\mu_p[\down{\beta}[\down{\alpha}[f]]]=\frac{1}{|\beta-2|^p} N_{1+(2-\beta)p}^{(\beta-2)p}[\down{\beta}[f]]=
\frac{\phi_{p(\beta-2),2-\alpha}^{p(\beta-2)(2-\alpha)}[f]}{|\beta-2|^{p}},
\end{equation}
leading to Eq.~\eqref{eqprop:downdown}.
\end{proof}
In the next theorem we relate the conditions under which a collection of numbers uniquely characterizes a probability density through Fisher-like measures with the Hausdorff moment problem.

\begin{theorem}
Let $\alpha<2, \beta<2$ be a pair of real numbers and $\{F_{p,\beta}\}_{p\in A}^\infty$ be a set of real values with $A=\{i(\alpha-2):\;i\in\mathbb N\}$. Let $\widetilde{\D}$ be the set of probability densities defined as
$$
\widetilde{\D}=\{g: g'<0,\;\text{and}\; gg''<\alpha\,(g')^2,\; \alpha<2\}.
$$
If the set of values $\{m_i\}_{i=1}^\infty$ such that
	\begin{equation}\label{eqproof:77}
	m_i=\frac{1}{|\alpha-2|^i}F_{i(\beta-2),2-\alpha}
	\end{equation}
is a moment-sequence of a probability density $f$ (i.e., $\mu_i[f]=m_i$) that uniquely characterizes $f$, then there is an unique probability density $g\in\widetilde \D$ such that $\phi_{p,\beta}^{p\beta}[g]=F_{p,\beta},$ for any $p\in A$. In particular, $g=\up{\alpha}[\up{\beta}[f]]$.
\end{theorem}
\begin{proof}
Let us denote by $\D$ the whole set of the probability density functions on $\mathbb R$. Let us suppose that $\{m_i\}_{i=1}^\infty$ is a moment sequence that uniquely characterizes certain probability density $f$ with $\mu_i[f]=m_i$. Then, from Propositions~\ref{prop:downdown} and~\ref{prop:inv}, we find
$$m_i=\mu_i[f]=\frac{\phi_{i(\beta-2),2-\alpha}^{i(\beta-2)(2-\alpha)}[f^{\uparrow\uparrow}_{\beta\alpha}]}{|\alpha-2|^{p}}.$$
Now, we note that one can define $\down{\beta}[\down{\alpha}[g]]$ if and only if $g\in\widetilde \D$ as discussed in Remark~\ref{rem:downcomposition}. Recalling that up transformed densities are well defined for any $f\in\D$, and the fact that $\up{\alpha}$ establishes a bijection between $\D$ and $\up{\alpha}[\D],$ it follows that $\up{\alpha}\up{\beta}$ also establishes a bijection between $\D$ and $\up{\alpha}[\up{\beta}[\D]]$. Then necessarily follows
$$
\up{\alpha}\up{\beta}[\D]=\widetilde \D,\qquad \text{or equivalently,}\qquad \down{\beta}\down{\alpha}[\widetilde \D]=\D,
$$
concluding the proof.
\end{proof}

\section{Stretched Gaussian densities and up and down transforms}

We compute in this section the outcome of the application of the down and up transformations to the minimizers $g_{p,\lambda}$ of the biparametric Stam and entropy-moment inequalities, defined in Eq. \eqref{def:g_plambda}. Together with the interest by themselves, these calculations will be essential for deducing the minimizers to some new, mirrored biparametric inequalities obtained through our transformations and presented in the following section. Note that $g_{p,\lambda}$ are generally defined as symmetric probability densitiy functions on $\Rset$; however in the following lines we will consider only their restriction to $\Rset_+$. The corresponding counterpart of the transformed density can be obtained by symmetry arguments.
\begin{proposition}\label{prop:downg}
	Let $\alpha\in\Rset$ and $\lambda>1-p^*$. We have
	\begin{enumerate}[a)]
		\item For $\lambda\neq 1$ and $\alpha\neq 2,$
		\begin{equation}\label{eq:down_gpl}
		\down{\alpha}[g_{p,\lambda}](s)=\frac{|1-\lambda|^{\frac 1p}}{a_{p,\lambda}^{\frac{\lambda-1}{p ^*}}p ^*}\left[(\alpha-2)s\right]^{\frac{\alpha+\lambda-2}{2-\alpha}}\left|\left[(\alpha-2)s\right]^{\frac{\lambda-1}{2-\alpha}}-a_{p,\lambda}^{\lambda-1}\right|^{-\frac {1}p}.
		\end{equation}
		\item For $\lambda=1$ and $\alpha\neq 2,$
		\begin{equation}\label{eq:down_gpl_l1}
		\down{\alpha}[g_{p,1}](s)=\frac{1}{p^*}[(\alpha-2) s]^{\frac{\alpha-1}{2-\alpha}}\left(\frac{\ln[a_{p,1}^{\alpha-2}(\alpha-2) s]}{\alpha-2}\right)^{-\frac{1}{p}}.
		\end{equation}
		\item For $\lambda\neq 1$ and $\alpha= 2,$
		\begin{equation}\label{eq:down_gpl_a2}
		\down{2}[g_{p,\lambda}]=\frac{|1-\lambda|^\frac 1p a_{p,\lambda}^{\frac 1{1-\lambda}}}{p^*} e^{-\lambda s}\left|a_{p,\lambda}^{-1}e^{(1-\lambda) s}-1\right|^{-\frac 1p}.
		\end{equation}
		\item For $\lambda= 1$ and $\alpha=2,$
		\begin{equation}\label{eq:down_gpl_a2l1}
		\down{2}[g_{p,1}]=\frac{e^{-s}}{p ^* (s+\ln a_{p,1})^\frac 1p}.
		\end{equation}
	\end{enumerate}
\end{proposition}
We observe that the family of functions obtained in Eq. \eqref{eq:down_gpl} corresponds to the probability density functions of the generalized Beta distributions.
\begin{proof}
	(a) Assume first that $\alpha\neq 2$ and $\lambda\neq 1$.  We start by computing $\down{\alpha} [g_{p,\lambda}]$, recalling at this point that $g_{p,\lambda}$ is defined for $x\in(0,\infty)$. On the one hand, its derivative is then given by
	\[|g_{p,\lambda}(x)'|=a_{p,\lambda}p ^*(1+(1-\lambda) x^{p^*})^{\frac{2-\lambda}{\lambda-1}}x^{\frac{1}{p-1}}.\]
	On the other hand, we deduce from Eq.~\eqref{eq:can_change} that
	\[s(x)=\frac{g_{p,\lambda}^{2-\alpha}(x)}{\alpha-2}=\frac{a_{p,\lambda}^{2-\alpha}}{\alpha-2} (1+(1-\lambda) x^{p^*})^{\frac{2-\alpha}{\lambda-1}},\]
	hence, inverting the latter change of variable
	\[(1-\lambda )x(s)^{p^*}=\left[\frac{(\alpha-2)s}{a_{p,\lambda}^{2-\alpha}}\right]^{\frac{\lambda-1}{2-\alpha}}-1.\]
	Note that, in the case $\lambda>1$, the support of $g_{p,\lambda}$ is compact, and one can easily deduce that the right hand side of the latter equation is always negative. We thus find that
	\[x(s)^{p^*}=\frac1{1-\lambda}\left(\left[\frac{(\alpha-2)s}{a_{p,\lambda}^{2-\alpha}}\right]^{\frac{\lambda-1}{2-\alpha}}-1\right)>0.\]
	We next compute $g_{p,\lambda}(x(s))$ and $g'_{p,\lambda}(x(s)).$ We have
	\begin{equation}\label{eq:proof1}
	g_{p,\lambda}(x(s))=a_{p,\lambda}(1+(1-\lambda)x(s)^{p^*})^{\frac1{\lambda-1}}=a_{p,\lambda}\left[\frac{(\alpha-2)s}{a_{p,\lambda}^{2-\alpha}}\right]^{\frac{1}{2-\alpha}}=\left[(\alpha-2)s\right]^{\frac{1}{2-\alpha}}
	\end{equation}
	and
	\begin{equation}\label{eq:proof2}
	\begin{split}
	|g'_{p,\lambda}(x(s))|&=a_{p,\lambda}p ^*(1+(1-\lambda) x(s)^{p^*})^{\frac{2-\lambda}{\lambda-1}}x(s)^{\frac{1}{p-1}}
	\\
	&=a_{p,\lambda}^{\lambda-1}p ^* \left[(\alpha-2)s\right]^{\frac{2-\lambda}{2-\alpha}} \left(\frac1{1-\lambda}\left(\left[\frac{(\alpha-2)s}{a_{p,\lambda}^{2-\alpha}}\right]^{\frac{\lambda-1}{2-\alpha}}-1\right)\right)^\frac 1p
	\\&=\frac{a_{p,\lambda}^{\frac{\lambda-1}{p ^*}}p ^*}{|1-\lambda|^{\frac 1p}} \left[(\alpha-2)s\right]^{\frac{2-\lambda}{2-\alpha}} \left|\left[(\alpha-2)s\right]^{\frac{\lambda-1}{2-\alpha}}-a_{p,\lambda}^{\lambda-1}\right|^\frac 1p.
	\end{split}
	\end{equation}
	Eq. \eqref{eq:down_gpl} then follows immediately from the definition Eq. \eqref{eq:down} of the down transformation and Eqs.~\eqref{eq:proof1} and ~\eqref{eq:proof2}.
	
	\smallskip
	(b) Let now $\lambda=1,\,\alpha\neq2$ and recall that $g_{p,1}(x)=a_{p,1}e^{-x^{p^*}}.$ We infer from the definition of $g_{p,1}$ that
	\[x(s)^{p^*}=\frac{\ln[a_{p,1}^{\alpha-2}(\alpha-2) s]}{\alpha-2}>0.\]
	We perform analogous calculations as above to find
	\begin{equation*}
	\begin{split}
	\down{\alpha}[g_{p,1}](s)&=\frac{g_{p,1}(x(s))^\alpha}{|g'_{p,1}(x(s))|}=\frac{a_{p,1}^\alpha e^{-\alpha x(s)^{p^*}}}{a_{p,1}p^*x(s)^{p^*-1}e^{-x(s)^{p^*}}}
	\\
	&=\frac{1}{p^*}[(\alpha-2) s]^{\frac{\alpha-1}{2-\alpha}}\left(\frac{\ln[a_{p,1}^{\alpha-2}(\alpha-2) s]}{\alpha-2}\right)^{-\frac{1}{p}},
	\end{split}
	\end{equation*}
	arriving thus at Eq. \eqref{eq:down_gpl_l1}.
	\smallskip
	
	(c) Fix now $\alpha=2$ and $\lambda\neq1$, recalling that $s(x)$ is given by Eq.~\eqref{eq:can_change_a2}. In particular, applying  Eq.~\eqref{eq:can_change_a2} to $g_{p,\lambda}$, we obtain
	\[s(x)=-\ln[g_{p,\lambda}(x)]=\frac{\ln[a_{p,\lambda} (1+(1-\lambda)x^{p^*})]}{1-\lambda},\]
	or equivalently
	\[x(s)^{p^*}=\frac1{1-\lambda}\left(a_{p,\lambda}^{-1}e^{(1-\lambda) s}-1\right)>0.\]
	Then, applying the definition of $\down{\alpha}$ transforms, we have
	\begin{equation*}
	\begin{split}
	\down{2}[g_{p,\lambda}]&=\frac{g_{p,\lambda}^2(x(s))}{|g_{p,\lambda}'(x(s))|}=\frac{a_{p,\lambda}^{\frac{2(2-\lambda)}{1-\lambda}} e^{-2s}}{a_{p,\lambda}^{\frac{2\lambda-3}{\lambda-1}}p ^* e^{(\lambda-2)s}\left[\frac1{1-\lambda}\left(a_{p,\lambda}^{-1}e^{(1-\lambda) s}-1\right)\right]^{\frac 1p}}
	\\
	&=\frac{|1-\lambda|^\frac 1p a_{p,\lambda}^{\frac 1{1-\lambda}}}{p^*} e^{-\lambda s}\left|a_{p,\lambda}^{-1}e^{(1-\lambda) s}-1\right|^{-\frac 1p},
	\end{split}
	\end{equation*}
	which is exactly Eq. \eqref{eq:down_gpl_a2}.
	\smallskip
	
	(d) We are finally left with the limiting case $\alpha=2$ and $\lambda=1.$ We find from Eq.~\eqref{eq:can_change_a2} applied to $g_{p,1}$ that
	\begin{equation}
	x(s)^{p^*}=s+\ln(a_{p,1}),
	\end{equation}
	and we further get
	\begin{equation*}
	\down{2}[g_{p,1}]=\frac{g_{p,1}^2(x(s))}{|g_{p,1}'(x(s))|}=\frac{e^{-2s}}{p ^* (s+\ln a_{p,1})^\frac 1p e^{-s}},
	\end{equation*}
	which leads to Eq. \eqref{eq:down_gpl_a2l1}.
\end{proof}

\begin{remark}
	As a consequence of the more general Remarks \ref{rem:can} and \ref{rem:downsupp}, since for the minimizers we always have $x_i=0$, $g_{p,\lambda}(0)=a_{p,\lambda}$ and $g_{p,\lambda}(x)\to0$ as $x\to x_f$ (where $x_f=\infty$ if $\lambda\leq1$ and $x_f\in(0,\infty)$ if $\lambda>1$ is the edge of the compact support in this latter case), we deduce that the support of $\down{\alpha}(g_{p,\lambda})$ is
	$$
	{\rm supp}\,\down{\alpha}(g_{p,\lambda})=\left\{\begin{array}{ll}\left[\frac{a_{p,\lambda}^{2-\alpha}}{\alpha-2},0\right], & {\rm if} \ \alpha<2,\\[3mm]
	\left[-\ln\,a_{p,\lambda},\infty\right), & {\rm if} \ \alpha=2,
	\\[3mm]
	\left[\frac{a_{p,\lambda}^{2-\alpha}}{\alpha-2},\infty\right), & {\rm if} \ \alpha>2.
	\end{array}\right.
	$$
\end{remark}

The next consequence of Proposition \ref{prop:downg} gathers some nice particular cases.
\begin{corollary}\label{cor:down_gpl}
		It follows directly from Eq.~\eqref{eq:down_gpl} that, for $\lambda\in\Rset\setminus\{0,1\}$,
\begin{equation*}
\down{2-\lambda}[g_{p,\lambda}]=g_{1-\lambda,1-p},\quad \text{or equivalently}, \quad \up{2-\lambda}[g_{1-\lambda,1-p}]=g_{p,\lambda}.
\end{equation*}
At the same time, this fact implies that
\begin{equation*}
\down{1+p}[g_{1-\lambda,1-p}]=g_{p,\lambda},\quad \text{or equivalently}, \quad \up{1+p}[g_{p,\lambda}]=g_{1-\lambda,1-p}.
\end{equation*}
We infer from the previous equalities that
\begin{equation*}
(\down{1+p}\down{2-\lambda})[g_{p,\lambda}]=g_{p,\lambda},\quad \text{or equivalently}, \quad (\up{2-\lambda}\up{1+p})[g_{p,\lambda}]=g_{p,\lambda}.
\end{equation*}
	\end{corollary}

The next result gathers the effect of the up transformation on the family of minimizers $g_{p,\lambda}$. We restrict the statement, in this case, to the parameters for which the outcome of the up transformation can be expressed in terms of well-established functions (such as the generalized trigonometric or hyperbolic ones), that is, the range $\lambda\neq1$ and $\alpha\in\Rset\setminus[1,2]$. A discussion of the remaining cases is included afterwards.

\begin{proposition}\label{prop:upg}
	Let $\alpha\in\Rset$ and $\lambda>1-p^*$. Then
	\begin{enumerate}[a)]
		\item For $\lambda>1$ and $\alpha\in\Rset\setminus[1,2]$ we have
		\begin{equation}\label{eq:up_gpl(>)}
		\up{\alpha}[g_{p,\lambda}](s)=\frac{|\alpha-2|^{\frac{1}{2-\alpha}}}{|\lambda-1|^{\frac{1}{(1-\alpha)p^*}}}
		\left|\sin_{1-\lambda,\frac{\alpha-2}{\alpha-1}p^*}\left(s_0-\frac{s}{C(\alpha,p,\lambda)}\right)\right|^{\frac{1}{1-\alpha}}.
		\end{equation}
		\item For $\lambda<1$ and $\alpha\in\Rset\setminus[1,2]$ we have
		\begin{equation}\label{eq:up_gpl(<)}
		\up{\alpha}[g_{p,\lambda}](s)=\frac{|\alpha-2|^{\frac{1}{2-\alpha}}}{|\lambda-1|^{\frac{1}{(1-\alpha)p^*}}}
		\left|\sinh_{1-\lambda,\frac{\alpha-2}{\alpha-1}p^*}\left(s_0-\frac{s}{C(\alpha,p,\lambda)}\right)\right|^{\frac{1}{1-\alpha}}.
		\end{equation}
	\end{enumerate}
\end{proposition}

\begin{proof}
We set $x_0\in(0,\infty)$ as integration endpoint throughout the proof to avoid problems with the definition in the change of variable in Eqs.~\eqref{eq:up} and \eqref{eq:up2}, recalling that the change of variables in the up transformation is defined up to a translation.
	
	(a) Assume first that $\lambda>1$ and $\alpha\in\Rset\setminus[1,2].$ The change of variables of the up transform~\eqref{eq:up} applied to the family of densities $g_{p,\lambda}$ is defined by
	\begin{equation}\label{eq:cambio}
	s(x)=\int_x^{x_0}|(\alpha-2)v|^\frac{1}{\alpha-2} g_{p,\lambda}(v)dv=a_{p,\lambda}|\alpha-2|^\frac{1}{\alpha-2}\int_x^{x_0}|v|^\frac{1}{\alpha-2} (1+(1-\lambda)v^{p^*})_+^{\frac{1}{\lambda-1}}dv.
	\end{equation}
We perform the change of variables
	
	$$w=\left(|\lambda-1|^{\frac1{p^*}} v\right)^{\frac{\alpha-1}{\alpha-2}}$$
	and we obtain after straightforward calculations
	\begin{equation}\label{eq:sin_gen}
	\begin{split}
	s(x)&=C(\alpha,p,\lambda)\int_{\xi(x)}^{\xi(x_0)} (1-w^{\frac{\alpha-2}{\alpha-1}p^*})^\frac{1}{\lambda-1}dw
	\\
	&=C(\alpha,p,\lambda)\left[\arcsin_{1-\lambda,\frac{\alpha-2}{\alpha-1}p^*}(\xi(x_0))-\arcsin_{1-\lambda,\frac{\alpha-2}{\alpha-1}p^*}(\xi(x))\right],
	\end{split}
	\end{equation}
	where
	\[C(\alpha,p,\lambda)=\frac{a_{p,\lambda}|\alpha-2|^{\frac1{\alpha-2}} (\alpha-2)}{|\lambda-1|^{\frac{\alpha-1}{p^*(\alpha-2)}}(\alpha-1)},\qquad \xi(x)=|\lambda-1|^{\frac1{p^*}}x^{\frac{\alpha-1}{\alpha-2}},\]
	and $\arcsin_{a,b}$ is the generalized arcsine function (see for example \cite{Drabek99, Yin19}), which is well defined since $\frac{\alpha-2}{\alpha-1}p^*>0$. Then, by inverting Eq.~\eqref{eq:sin_gen}, we arrive to
	$$
	x(s)=\left[\frac{\sin_{1-\lambda,\frac{\alpha-2}{\alpha-1}p^*}\left(s_0-\frac{s}{C(\alpha,p,\lambda)}\right)}{|\lambda-1|^ {\frac{1}{p^*}}}\right]^{\frac{\alpha-2}{\alpha-1}}, \quad
	s_0=\arcsin_{1-\lambda,\frac{\alpha-2}{\alpha-1}p^*}(\xi(x_0)),
	$$
	which immediately implies Eq. \eqref{eq:up_gpl(>)}.
	
	\smallskip
	
	(b) Let us notice that, for $\lambda<1$, the calculations leading to Eq. \eqref{eq:sin_gen} are completely analogous, with the single difference than in the integrand contained in the right hand side of the expression of $s(x)$ we have $1+w^{\frac{\alpha-2}{\alpha-1}p^*}$ (that is, the minus sign changes into a plus sign), which leads to a closed expression similar to Eq. \eqref{eq:up_gpl(>)} but with the hyperbolic counterpart of the generalized sine function (see for example \cite{Yin19}).
\end{proof}

\begin{remark}
	Taking into account that the definition of $\arcsin_{v,w}$, we observe that
	$$
	\sin_{v,1}(x)=1-\left(1-\frac{v-1}v x\right)^\frac{v}{v-1}.
	$$
	Thus, if we let $\alpha=1+p$, which is equivalent to
	$$
	\frac{\alpha-2}{\alpha-1}p^*=1,
	$$
	in Eq.~\eqref{eq:sin_gen}, we recover the stretched Gaussian densities $g_{1-\lambda,1-p}$ (as mentioned in Corollary~\ref{cor:down_gpl}) up to a translation.
\end{remark}

\noindent \textbf{Discussion of the case $\alpha\in[1,2]$, $\lambda\neq1$.} Let us point out here that, for $\alpha\in(1,2)$ and $\lambda\neq1,$ Eq.~\eqref{eq:cambio} still makes sense, despite the negativity of $\frac{\alpha-2}{\alpha-1}p^*$, while the change of variable leading to Eq. \eqref{eq:sin_gen} is no longer well defined due to the fact that the generalized \textit{arcsine} function is defined only for positive values of its parameters. Thus, we cannot write the transform in terms of well-established special or generalized functions. In the limiting case $\alpha=1$ we have to introduce the change of variable $w=\ln\,v$, which gives
\begin{equation}\label{eq:a1}
s(x)=a_{p,\lambda}\int_{\ln\,x}^{\ln\,x_0}(1+(1-\lambda)e^{p^*w})_+^{\frac{1}{\lambda-1}}dw.
\end{equation}
In both cases $\alpha=1$ or $\alpha\in(1,2)$, by definition, the up transform is given by
\begin{equation}
\up{\alpha}[g_{p,\lambda}](s)=|(\alpha-2)x(s)|^\frac{1}{2-\alpha}
\end{equation}
and it can still be applied starting from Eq. \eqref{eq:a1}, respectively Eq. \eqref{eq:cambio}, but without a "nice" expression as the ones given in the statement of Prop.~\eqref{prop:upg}. Finally, for $\alpha=2$ and $\lambda\neq1$, one finds
\begin{equation}
\up{2}[g_{p,\lambda}](s)=e^{-x(s)},\quad s(x)=\int_x^{x_0}e^vg_{p,\lambda}(v)dv.
\end{equation}

Let us now complete our analysis with the case $\lambda=1$. We have
\begin{proposition}\label{prop:upl1}
	Let $\alpha\in\Rset\setminus[1,2]$ and $\lambda=1$. Then
	\begin{equation}\label{eq:upl1}
	\up{\alpha}[g_{p,1}](s)=|\alpha-2|^{\frac{1}{\alpha-2}}\left|\Gamma^{-1}\left(\frac{\alpha-1}{(\alpha-2)p^*},\frac{s+s_0}{K}\right)\right|^{\frac{1}{p^*(\alpha-2)}},
	\end{equation}
	where
	\begin{equation}\label{eq:notupl1}
	K=\frac{a_{p,1}}{p^*}|\alpha-2|^{\frac{1}{\alpha-2}}, \quad s_0=K\Gamma\left(\frac{\alpha-1}{(\alpha-2)p^*},x_0^{p^*}\right).
	\end{equation}
\end{proposition}
\begin{proof}
	Recalling that $g_{p,1}(x)=a_{p,1}e^{x^{p^*}}$, we obtain after some algebraic manipulations that
	\begin{equation}\label{eq:proofupl1}
	\begin{split}
	s(x)&=\int_x^{x_0}|(\alpha-2)v|^{\frac{1}{\alpha-2}}g_{p,1}(v)dv=a_{p,1}|\alpha-2|^{\frac{1}{\alpha-2}}\int_x^{x_0}v^{\frac{1}{\alpha-2}}e^{-v^{p^*}}dv\\
	&=\frac{a_{p,1}}{p^*}|\alpha-2|^{\frac{1}{\alpha-2}}\int_{x^{p^*}}^{x_0^{p^*}}w^{\frac{p-\alpha+1}{(\alpha-2)p}}e^{-w}dw\\
	&=\frac{a_{p,1}}{p^*}|\alpha-2|^{\frac{1}{\alpha-2}}\left[\Gamma\left(\frac{\alpha-1}{(\alpha-2)p^*},x^{p^*}\right)-\Gamma\left(\frac{\alpha-1}{(\alpha-2)p^*},x_0^{p^*}\right)\right],
	\end{split}
	\end{equation}
	where the notation involving the Gamma function is justified in the range where the Gamma function is defined, which corresponds to $\alpha\in\Rset\setminus[1,2]$, as assumed. Recalling the notation introduced in Eq. \eqref{eq:notupl1}, we further obtain from the latter identity that
	$$
	\Gamma\left(\frac{\alpha-1}{(\alpha-2)p^*},x^{p^*}\right)=\frac{s+s_0}{K},
	$$
	which, by inverting the incomplete Gamma function and taking powers $1/p^*$, readily leads to Eq. \eqref{eq:upl1}.
\end{proof}

\noindent \textbf{The range $\alpha\in[1,2)$ and $\lambda=1$.} In this case, the Gamma function is no longer well defined and thus we cannot introduce the notation Eq. \eqref{eq:notupl1} and express the result as in Eq. \eqref{eq:upl1}. However, the calculations leading to Eq. \eqref{eq:proofupl1} are still correct and we can proceed by inverting $s(x)$, but without reaching a closed expression in terms of Gamma functions.

\smallskip

\noindent \textbf{The limiting case $\alpha=2$ and $\lambda=1$.} In this case, we obtain
\begin{equation}
\up{2}[g_{p,\lambda}](s)=e^{-x(s)},\quad s(x)=a_{p,\lambda}\int_x^{x_0} e^{v-v^{p^*}}dv,
\end{equation}
which allows us to define the up transformation, but it apparently cannot be expressed in terms of well-established special functions, as we did in Proposition \ref{prop:upl1}.

\section{Mirrored moment-entropy and Stam like inequalities}

Once introduced the up and down transformations and explored their basic properties related to the supports and monotonicity of the transformed functions, we study in this section the effect of their application in connection with the Stam inequality, the moment-entropy inequality and their minimizers. We thus establish new informational inequalities, holding true in ranges of exponents that are not taken into account in the corresponding classical inequalities and whose minimizers are quite different with respect to the functions $g_{p,\lambda}$ defined in Eq. \eqref{def:g_plambda}. In order to simplify the statements of the forthcoming theorems, we introduce the following notation.

\begin{definition}\label{def:gpbl}
Given $p,\beta,\lambda$ such that
\begin{equation}\label{eq:cond_mirrored_gpbl}
\sign(\lambda-1-p\beta)=\sign(\lambda-1)=\sign (\beta)\neq 0,\quad {\rm if} \ p<1
\end{equation}
and
\begin{equation}\label{eq:cond_mirrored_gpbl2}
\sign\left(\frac{1-\lambda}{p}+\beta\right)=\sign\left(1- \lambda+\beta\right)\neq0,\quad {\rm if} \ p\geqslant1,
\end{equation}
we define the following probability densities
\begin{equation}\label{def:tilde_gpbl}
\widetilde g_{p,\beta,\lambda} =\left\{\begin{array}{ll}\overline  g_{p,\beta,\lambda}=\up{2-\beta}[g_{\frac{\lambda-1}{\lambda-1-\beta},1-p}], &  p<1,
\\[2mm]
g_{p,\beta,\lambda}, &  p\geqslant 1.\end{array}\right.
\end{equation}
\end{definition}
\begin{remark}
Note that the conditions $p^*>0$ and $\lambda>1-p^*$ needed for the integrability of $g_{p,\lambda}$ turn into Eq.~\eqref{eq:cond_mirrored_gpbl} by identifying $p=\frac{\lambda-1}{\lambda-\beta-1}$ and $\lambda=1-p$. The condition~\eqref{eq:cond_mirrored_gpbl2} appears in the definition of $g_{p,\beta,\lambda}$ (see~\cite[Section 4.1]{Puertas25}).
\end{remark}

\subsection{Extended Stam inequalities}

The goal of this section is to state and prove a generalization of the Stam inequality Eq. \eqref{ineq:bip_Stam}. This generalization acts in two directions: a first extension, still with $p\geq1$, involves three parameters, while in a second step we also extend the range of the parameters of the first generalized inequality to include exponents $p<1$.
\begin{theorem}[Generalized Stam inequalities in an extended domain]\label{th:gener_Stam}
Let $p\geqslant 1,$ and $\beta$ be such that
\begin{equation}\label{eq:sign_cond1}
\sign\left(p^*\beta + \lambda-1\right)=\sign\left(\beta+1-\lambda\right)\neq0.
\end{equation}
Then, the following generalized Stam inequality holds true for $f: \Rset\mapsto\Rset^+$ absolutely continuous if $1+\beta-\lambda>0$ or for $f:(x_i,x_f)\mapsto \Rset^+$ absolutely continuous on $(x_i,x_f)$ if $1+\beta-\lambda<0$:
\begin{equation}\label{ineq:trip_Stam_extended}
	\left(\phi_{p,\beta}[\pdf] \, N_{\lambda}[\pdf] \right)^{\theta(\beta,\lambda)}\: \geqslant \:	\left(\phi_{p,\beta}[\widetilde g_{p,\beta,\lambda}] \, N_{\lambda}[\widetilde g_{p,\beta,\lambda}] \right)^{\theta(\beta,\lambda)}\equiv \kappa^{(1)}_{p,\beta,\lambda},
\end{equation}
where $\theta(\beta,\lambda)=1+\beta-\lambda$ and $\widetilde g_{p,\beta,\lambda}$ is defined in Eq. \eqref{def:tilde_gpbl}. Moreover, for $p<1$ the inequality ~\eqref{ineq:trip_Stam_extended} still holds true with $\theta(\beta,\lambda)=-\beta$, provided that $f:\Rset \mapsto\Rset^+$ is continuously differentiable with $f'<0$ and the domain of its exponents is limited by the conditions
\begin{equation}\label{eq:sign_cond2}
\sign(p^*\beta+\lambda-1)=\sign(\beta+\lambda-1)\neq0,\qquad \sign(\lambda-1)=\sign(\beta)\neq0.
\end{equation}

The optimal lower bound in the generalized Stam inequality Eq. \eqref{ineq:trip_Stam_extended} is given by
	\begin{equation}\label{eq:opt_const_gener_Stam}
	\kappa^{(1)}_{p,\beta,\lambda}=\left\{\begin{array}{ll}
	\bigg(K^{(1)}_{p,\beta,\lambda}\bigg)^{1+\beta-\lambda}
	=|1+\beta-\lambda|^{-\frac{1+\beta-\lambda}{\beta}}K^{(1)}_{p,\frac{\beta}{1+\beta-\lambda}}, & {\rm for} \ p\geqslant1,\\[5mm] |\beta|K^{(0)}_{\frac{\lambda-1}{\lambda-1-\beta},1-p}, & \rm{for} \ p<1,\end{array}\right.
	\end{equation}
	where $K^{(0)}_{p,\lambda},$ $K^{(1)}_{p,\lambda}$  and $K^{(1)}_{p,\beta,\lambda}$ are the optimal constants of the classical inequalities in Eqs. \eqref{ineq:bip_E-M}, \eqref{ineq:bip_Stam} and Eq.  \eqref{ineq:trip_Stam}.

	\end{theorem}
\begin{proof}
Assume first $p\geqslant1$. As discussed in the Introduction, given $\lambda$ such that
\begin{equation}\label{eq:init_cond}
\lambda>\frac1{1+p^*},
\end{equation}
and any absolutely continuous probability density function $\pdf$, the two-parameter Stam inequality Eq. \eqref{ineq:bip_Stam} is fulfilled, see for example \cite{Lutwak05,Bercher12}. Recalling the differential-escort transformations \cite{Zozor17,Puertas19} and applying Eq. \eqref{ineq:bip_Stam} to a differential-escort transformed density, we obtain, for any $\alpha\neq0$, that
\begin{equation}\label{eq:StamEscort}
\begin{split}
N_\lambda[\pdf]\phi_{p,\lambda}[\pdf]&=N_\lambda[\mathfrak E_\alpha\,\mathfrak E_{\alpha^{-1}}[\pdf]]\phi_{p,\lambda}[\mathfrak E_\alpha\mathfrak E_{\alpha^{-1}}[\pdf]]\\
&=|\alpha|^\frac1{\lambda}\left(N_{1+\alpha(\lambda-1)}[\mathfrak E_{\alpha^{-1}}[f]]\phi_{p,\alpha\lambda}[\mathfrak E_{\alpha^{-1}}[f]]\right)^\alpha\geqslant K^{(1)}_{p,\lambda}.
\end{split}
\end{equation}
Introducing the following notation
\begin{equation}\label{eq:not_bar}
\overline \lambda=1+\alpha(\lambda-1),\quad \overline\beta=\alpha\lambda,
\end{equation}
which is equivalent to
$$
\alpha=1+\overline\beta-\overline\lambda\neq0, \quad \lambda=\frac{\overline\beta}{1+\overline\beta-\overline\lambda},
$$
we infer from Eq. \eqref{eq:StamEscort} that

\begin{equation}\label{ineq:gener_Stam(1)}
\left(N_{\overline\lambda}[\overline\pdf]\phi_{p,\overline{\beta}}[\overline\pdf]\right)^{1+\overline\beta-\overline\lambda} \geqslant|\alpha|^{-\frac{1}{\lambda}}K^{(1)}_{p,\lambda}
=|1+\overline\beta-\overline\lambda|^{\frac{\overline\lambda-\overline\beta-1}{\overline\beta}}K^{(1)}_{p,\frac{\overline\beta}{1+\overline\beta-\overline\lambda}},
\end{equation}
where $\overline f=\mathfrak E_{\alpha^{-1}}[f]=\mathfrak E_{\frac{1}{1+\overline \beta-\overline{\lambda}}}[f].$
Note that for $\alpha>0$, $\overline f$ remains absolutely continuous, while for $\alpha<0$ the density $\overline f$ is divergent in the borders.

In order to remain in the framework of the inequality Eq. \eqref{ineq:bip_Stam}, we have to check the condition~\eqref{eq:init_cond}. In particular, the second of these conditions is fulfilled if
\begin{equation*}
\frac{\overline\beta}{1+\overline\beta-\overline\lambda}-\frac1{1+p^*}>0,
\end{equation*}
that is,
\begin{equation*}
\frac{\overline\beta p^*-1+\overline\lambda}{1+\overline\beta-\overline\lambda}>0,
\end{equation*}
which is equivalent to Eq.~\eqref{eq:sign_cond1}.

Let us now consider the mirrored case $p<1$. We infer from the generalized moment-entropy inequality Eq. \eqref{ineq:bip_E-M}, which we apply to a density obtained via a down transformation with $\alpha\in\Rset\setminus\{2\}$, the following inequality
\begin{equation}\label{eq:EMdown}
\frac{\sigma_{p^*}[\down{\alpha}[\pdf]]}{N_\lambda[\down{\alpha}[\pdf]]}\geqslant K^{(0)}_{p,\lambda},\quad p^*>0,\quad \lambda>\frac{1}{1+p^*}, \quad \lambda\neq1.
\end{equation}
Combining Eq. \eqref{eq:EMdown} with the equalities in Lemma \ref{lem:MEF}, we obtain the following generalized Stam inequality:
\begin{equation}\label{ineq:gener_Stam}
\left(N_{1+(2-\alpha)p^*}[\pdf]\phi_{1-\lambda,2-\alpha}[\pdf]\right)^{\alpha-2}\geqslant |2-\alpha|K^{(0)}_{p,\lambda}.
\end{equation}
Introducing the notation
\begin{equation}\label{eq:not_tilde}
\widetilde{\lambda}=1+(2-\alpha)p^*\neq 1,\quad \widetilde p=1-\lambda,\quad \widetilde\beta=2-\alpha\neq 0,
\end{equation}
the inequality Eq. \eqref{ineq:gener_Stam} writes in the equivalent form
\begin{equation}\label{ineq:gener_Stam(2)}
\left(N_{\widetilde \lambda}[\pdf]\phi_{\widetilde{p},\widetilde \beta}[\pdf]\right)^{-\widetilde\beta}\geqslant |\widetilde\beta|K^{(0)}_{p,\lambda}=|\widetilde\beta|K^{(0)}_{\frac{\widetilde \lambda-1}{\widetilde \lambda-1-\widetilde \beta},1-\widetilde p}.
\end{equation}
%
%
To determine the domain of the parameters $(\widetilde{p}, \widetilde\lambda, \widetilde\beta)$, taking into account that $p^*>0$, we get
\[\sign(\widetilde\lambda-1)=\sign(2-\alpha)=\sign(\widetilde\beta)\neq0,\]
which is the second condition in Eq.~\eqref{eq:sign_cond2}. Moreover, the condition in Eq. \eqref{eq:init_cond} is equivalent, in the new notation Eq. \eqref{eq:not_tilde} (and taking into account that $p^*=\frac{\widetilde \lambda-1}{\widetilde \beta}$), to
\[\widetilde p<1-\frac{1}{1+p^*}=\frac{p^*}{1+p^*}=\frac{\widetilde \lambda-1}{\widetilde \beta+\widetilde \lambda-1}<1,\]
that is,
\[\frac{(\widetilde \lambda-1)(1-\widetilde p)-\widetilde p \widetilde \beta}{\widetilde \beta+\widetilde \lambda-1}>0.\]
By dividing by $1-\widetilde p>0$ in previous inequality, we are left with
\[\frac{\widetilde \lambda-1+\widetilde p^* \widetilde \beta}{\widetilde \beta+\widetilde \lambda-1}>0,\]
which is equivalent to the first condition in Eq.~\eqref{eq:sign_cond2}.

The inequalities~\eqref{eq:EMdown} and~\eqref{ineq:gener_Stam} reach their minimum value when $\down{\alpha}[f]=g_{p,\lambda},$  or equivalently, $f=\up{\alpha}[g_{p,\lambda}].$ Taking into account the notaton in Eq.~\eqref{eq:not_tilde}, one finds after simple algebraic manipulations that the minimizing densities of the inequality~\eqref{ineq:gener_Stam(2)} are given by the densities $\widetilde g_{p,\beta,\lambda}$ defined in Eq~\eqref{def:tilde_gpbl}. Finally, the expression of the optimal constant $K^{(1)}_{p,\beta,\lambda}$ in both ranges $p\geq1$ and $p<1$ follows readily from Eqs. \eqref{ineq:gener_Stam(1)}, respectively \eqref{ineq:gener_Stam(2)}.
\end{proof}
%

%
\begin{remark}
In order to deduce a similar inequality in the case $\alpha=2$, one needs $\beta=0$ (see Eq.~\eqref{eq:EF}), which is not allowed in the definition of $\phi_{p,\beta}$ (see \cite{Lutwak05}). However, for $p\geqslant1$ and $\lambda>\frac{1}{1+p^*}$, we apply Eqs. \eqref{eq:Sp_down2} and \eqref{eq:Ren_down2} to deduce that
\begin{equation}\label{ineq:a2}
\frac{\sigma_{p^*}[\down2[f]]}{N_{\lambda}[\down2[f]]}=
\overline{S}_{p^*}[f]F_{1-\lambda,0}[f]^{\frac{1}{\lambda-1}}\geqslant K^{(0)}_{p,\lambda}.
\end{equation}
\end{remark}

\begin{remark}
We observe that, if we restrict ourselves to $p<0$, which gives $p^*\in(0,1)$, then Eq. \eqref{eq:not_tilde} gives
\[0<p^*=\frac{\widetilde \lambda-1}{\widetilde \beta}-1=\frac{\widetilde \lambda-1-\widetilde \beta}{\widetilde \beta},\]
that is,
\[\sign(\widetilde \beta)=\sign(\widetilde \lambda-1-\widetilde \beta).\]
Then, the condition Eq. \eqref{eq:sign_cond2} rewrites exactly as the opposite of Eq.~\eqref{eq:sign_cond1}, more precisely
\[\sign(\widetilde \lambda-1+\widetilde p^* \widetilde \beta)=\sign(\widetilde \lambda-1-\widetilde \beta).\]
\end{remark}

\begin{remark}
If we take $\alpha=2+\frac{1}{p^*-1}=p+1$, we infer from Eq. \eqref{ineq:gener_Stam(2)} that
\begin{equation}\label{eq:genStam_part}
|p^*-1|\left(N_{\frac{1}{1-p^*}}[\pdf]\phi_{1-\lambda,\frac{1}{1-p^*}}[\pdf]\right)^{\frac{1}{p^*-1}}\geqslant K^{(0)}_{p,\lambda}.
\end{equation}
The minimizing densities of the inequality Eq. \eqref{eq:genStam_part} are then given by $g_{1-\lambda,1-p},$ hence we see that
\[\up{1+p}[g_{p,\lambda}]=g_{1-\lambda,1-p}.\]
A remarkable feature of the previous equality is that, if we apply the up transformation $\up{2-\lambda}$ to $g_{1-\lambda,1-p}$, we recover the original stretched Gaussian density $g_{p,\lambda}.$
\end{remark}

We end this section by introducing a new sharp inequality dealing with Rényi and Shannon entropies, which is obtained through a similar approach, but applied to the special case $\lambda=1$. In this case, we have:
\begin{proposition}\label{prop:Shannon}
	Let $\alpha\in\Rset\setminus\{2\}$, $p\geqslant 1$ and $f$ be a derivable and decreasing probability density function. Then we have
	\begin{equation}\label{ineq:newLMC}
	N_{1+(2-\alpha)p}^{\alpha-2}[f]e^{-\alpha S[f]}\ge \mathfrak K e^{\left\langle\log\left|f'\right|\right\rangle},\qquad \mathfrak K=|2-\alpha| K^{(0)}_{p,\lambda}.
	\end{equation}
Moreover, the equality only holds when $f=\up{\alpha}[g_{p,1}],$ given in Proposition~\ref{prop:upl1}.
\end{proposition}
\begin{proof}
For $p\geqslant1$, we start from
$$
\frac{\sigma_{p^*}[\adown]}{e^{S[\adown]}}\ge K^{(0)}_{p,1}
$$
and the inequality \eqref{ineq:newLMC} follows directly by substitution from Lemma \ref{lem:MEF}. The equality in the latter inequality only holds true when $\adown=g_{p,1}.$ We leave the easy details to the reader.
\end{proof}

Note that we do not obtain a Stam-like inequality, since the Fisher information is not present, but we instead arrive to a different functional involving the derivative of the density. Actually, Eq. \eqref{ineq:newLMC} reminds of the inequality
$$
N_\beta[f]e^{-S[f]}\ge1, \quad \beta<1,
$$
optimized by the uniform density, but Eq.~\eqref{ineq:newLMC} is only valid for derivable probabilities with $f'<0$.

\subsection{Mirrored domain for the moment-entropy inequality}

We recall here that, when $p^*>0$ and $\lambda>\frac1{1+p^*}$, the moment-entropy inequality~\eqref{ineq:bip_E-M} is fulfilled, according to \cite{Lutwak05}. Our next theorem extends the range of this inequality to negative values of $\lambda$, with the price of taking a power and provided that the involved quantities are well defined.
\begin{theorem}[Entropy-momentum inequalities in an extended domain]\label{th:gener_EM}
Let $\lambda<0$ and $p$ be such that
\begin{equation}\label{eqtheo:conds_EM}
\sign\left(\frac{\lambda-1}{ \lambda}+ p^*\right)=\sign\left(1-p^*\right).
\end{equation}
Then, for any continuously differentiable density function $f$, the following mirrored moment-entropy inequality holds true:
\begin{equation}\label{ineq:bip_E-M_mirrored}
\left(\frac{\sigma_{p^*}[\pdf]}{|p^*-1|N_{\lambda}[\pdf]}\right)^{p^*-1}\geqslant\left(\frac{\sigma_{p^*}[\overline g_{p,\lambda}]}{|p^*-1|N_{\lambda}[\overline g_{p,\lambda}]}\right)^{p^*-1}\equiv  \kappa^{(0)}_{p,\lambda},
\end{equation}
where $\overline g_{p,\lambda}=g_{1-\lambda,1-p}$ and $\kappa^{(0)}_{p,\lambda}= K^{(1)}_{1-\lambda,\frac 1{1-p}}.$
\end{theorem}
\begin{proof}
Let us recall, as a starting point for our next inequalities, that for any $\lambda\neq1$, $\alpha\neq2$ and $\beta\neq0$, Lemma \ref{lem:MEF} ensures that
$$
N_{\lambda}[\aup]=\left(|2-\alpha| \sigma_{\frac{\lambda-1}{2-\alpha}}[\pdf]\right)^{\frac{1}{\alpha-2}},\quad
\phi_{p,\beta}[f^{\uparrow}_{2-\beta}]=\left(N_{1-p}[\pdf]\right)^{\frac1{\beta}}.
$$
Particularizing the previous equalities for $\beta=2-\alpha\neq0$, we find
\begin{equation}\label{eqproof:45}
N_{\lambda}[\aup]\phi_{p,2-\alpha}[f^{\uparrow}_{\alpha}]
=\left(|2-\alpha|\sigma_{\frac{\lambda-1}{2-\alpha}}[\pdf]\right)^{\frac{1}{\alpha-2}}\left(N_{1-p}[\pdf]\right)^{\frac1{2-\alpha}}
=\left(\frac{|2-\alpha|\sigma_{\frac{\lambda-1}{2-\alpha}}[\pdf]}{N_{1-p}[\pdf]}\right)^{\frac{1}{\alpha-2}}.
\end{equation}
Let us introduce the notation
\begin{equation}\label{eqproof:99}
\widehat p^*=\frac{\lambda-1}{2-\alpha}=\frac{\lambda-1}{\beta},\quad\text{and}\quad \widehat\lambda=1-p.
\end{equation}
In this notation and recalling that $\beta=2-\alpha$, we deduce on the one hand from Eq.~\eqref{eq:sign_cond1} that
$$
\sign\left((2-\alpha)\left(p^*+\frac{\lambda-1}{2-\alpha}\right)\right)=\sign\left((2-\alpha)\left(1+\frac{1-\lambda}{2-\alpha}\right)\right),
$$
or equivalently,
$$
\sign(\widehat p^*+(1-\widehat\lambda)^*)=\sign(1-\widehat p^*),
$$
leading to Eq. \eqref{eqtheo:conds_EM}. On the other hand, starting from Eq.~\eqref{eqproof:45} and taking into account the generalized Stam inequality Eq.~\eqref{ineq:trip_Stam_extended}, we obtain that
\begin{equation}\label{eqproof:10}
\left(\frac{|\beta|\sigma_{\frac{\lambda-1}{\beta}}[\pdf]}{N_{1-p}[\pdf]}\right)^{-\frac{\theta(\beta,\lambda)}{\beta}}
\geqslant (K^{(1)}_{p,\beta,\lambda})^{1+\beta-\lambda}.
\end{equation}
Since, by Eq. \eqref{eqproof:99}, we have
$$
\widehat p^*-1=\frac{\lambda-1-\beta}{\beta}=-\frac{\theta(\beta,\lambda)}{\beta},
$$
we obtain from Eq.~\eqref{eqproof:10} that
	\begin{equation}\label{eqproof:100}
		\left(\frac{|\beta|\sigma_{\widehat p ^*}[\pdf]}{N_{\widehat \lambda}[\pdf]}\right)^{\widehat p ^*-1}
		\geqslant (K^{(1)}_{p,\beta,\lambda})^{1+\beta-\lambda}.
	\end{equation}
Observe that the quantity (see Eq.~\eqref{ineq:trip_Stam})
$$ \frac1{|\beta|}(K^{(1)}_{p,\beta,\lambda})^{\frac{1+\beta-\lambda}{\widehat p ^*-1}}=\frac1{|\beta|}(K^{(1)}_{p,\beta,\lambda})^{-\beta}
=\left|\frac{\beta+1-\lambda}{\beta}\right| \bigg(K^{(1)}_{p,\frac{\beta}{\beta+1-\lambda}}\bigg)^{\frac{-\beta}{\beta+1-\lambda}}
$$
only depends on $\widehat p^*$ and $\widehat\lambda$ . We next choose $\beta=\lambda$, hence $|\beta|=\frac{1}{|\widehat p^*-1|}$ and $K^{(1)}_{p,\beta,\lambda}=K^{(1)}_{p,\lambda}$ and then
Eq.~\eqref{ineq:bip_E-M_mirrored} follows readily from Eq. \eqref{eqproof:100}, together with its optimal constant. Indeed, the claimed expression of the optimal constant $\kappa^{(0)}_{p,\lambda}$ follows from the optimal constant given in Eq. \eqref{eq:opt_const_gener_Stam} together with the notation for $\theta(\beta,\lambda)$ in the statement of Theorem \ref{th:gener_Stam} and the notation in Eq. \eqref{eqproof:99}. Moreover, the minimizing densities $f_{min}$ of the inequality~\eqref{ineq:bip_E-M_mirrored} satisfy
\begin{equation}\label{eqproof:88}
(f_{min})^\uparrow_{2-\beta}=g_{p,\beta,\lambda}.
\end{equation}
This follows from the fact that, according to Eq.~\eqref{eqproof:45}, $(f_{min})^\uparrow_\alpha$ are the minimizing densities of the triparametric Stam inequality~\eqref{ineq:trip_Stam} with $\beta=2-\alpha$. Recalling the choice $\beta=\lambda$, one obtains from Proposition~\ref{prop:inv} and Corollary~\ref{cor:down_gpl} that
\begin{equation}
f_{min}=\down{2-\lambda}[g_{p,\lambda}]=g_{1-\lambda,1-p},
\end{equation}
completing the proof.
\end{proof}


\begin{corollary}
Given $p,\lambda,\beta$ and $\alpha$ real numbers, provided the densities are well-defined, we have
\[\down{\alpha}[g_{p,2-\alpha,\lambda}]=\down{\beta}[g_{p,2-\beta,\lambda}].\]
or equivalently
\[g_{p,\beta,\lambda}=(\up{2-\beta}\down{2-\alpha})[g_{p,\alpha,\lambda}]=(\up{2-\beta}\down{2-\lambda})[g_{p,\lambda}].\]
\end{corollary}
\begin{proof}
The proof follows from Eq. \eqref{eqproof:88} and Proposition~\ref{prop:inv}.	
\end{proof}

\begin{remark}
An analogous inequality to Eq. \eqref{ineq:a2} follows also in the mirrored domain by applying \eqref{ineq:bip_E-M_mirrored} to the density $\down2[f]$. More precisely, employing once more Eqs. \eqref{eq:Sp_down2} and \eqref{eq:Ren_down2}, we readily obtain that
$$
\left[\overline{S}_{p^*}[f]F_{1-\lambda,0}[f]^{\frac{1}{\lambda-1}}\right]^{p^*-1}\geqslant |p^*-1|^{p^*-1}\kappa^{(0)}_{p,\lambda},
$$
for $\lambda<0$ and $p$ satisfying Eq. \eqref{eqtheo:conds_EM}.
\end{remark}	

We recall that in the usual case, a biparametric extension of the Crámer-Rao inequality can be derived by simply multiplying the moment-entropy and Stam inequalities~\cite{Lutwak04}. In contrast to this, in the mirrored domain, a closer inspection of the allowed values of the parameters $p, \lambda$ in the inequalities \eqref{ineq:trip_Stam_extended} and \eqref{ineq:bip_E-M_mirrored} shows the non-existence of a mirrored version of the Crámer-Rao inequality, observing that the exponents in the above mentioned inequalities have opposite signs.

\section{Conclusions}

In this work we have introduced a pair of transformations which are mutually inverse, called the up/down transformations respectively. We have analyzed their basic mathematical properties and studied their relation with the main informational measures. On the one hand, we have shown that the absolute $p$-moment of the down transformed density turns out to be essentially the Rényi entropy (or $L_p$ norm) of the original one. On the other hand, the Rényi entropy of the down transformed density is shown to be equal to an extension of the Fisher information of the original probability density. Moreover, we have used these results together with well-established extensions of the classical Stam and moment-entropy inequalities to extend them to a mirrored domain of the parameters. We have shown that, in this mirrored domain, the minimizing densities of our new inequalities are unbounded probability density functions. The optimal bounds and the minimizing densities are explicitly computed. Furthermore, the classical Cramér-Rao inequality cannot be extended to the mirrored range of parameters.
	
A further application of the up and down transformations has been found in the framework of the Hausdorff moment problem. More precisely, these transformations allow us to connect the Hausdorff entropic moment problem with the classical one and, furthermore, we introduce a Fisher-type moment problem and investigate its connection to the Hausdorff moment problem.
	
Further investigation based on introducing new informational functionals related to the up and down transformed densities will be carried out in a forthcoming work.


\subsection*{Acknowledgements}
	
R. G. I. and D. P.-C. are partially supported by the Grants PID2020-115273GB-I00 and RED2022-134301-T funded by \text{MCIN/AEI/10.13039/501100011033}. D. P.-C. is also partially supported by the Grant PID2023-153035NB-100.

\bigskip

\noindent \textbf{Data availability} Our manuscript has no associated data.

\bigskip

\noindent \textbf{Competing interest} The authors declare that there is no competing interest.


\bibliographystyle{unsrt}
\bibliography{refs}
\end{document}